\documentclass[lettersize,journal]{IEEEtran}
\usepackage{algorithm}
\usepackage{algpseudocode}

\usepackage{adjustbox}
\usepackage{algpseudocode}
\usepackage{graphics}
\usepackage{epsfig}
\usepackage{amsmath,amssymb,amsfonts}
\usepackage{amsthm}
\usepackage{graphicx}
\usepackage{textcomp}
\usepackage{xcolor}
\usepackage{flushend}
\usepackage{threeparttable}
\usepackage{verbatim}
\usepackage{enumitem}
\usepackage{graphicx,subfig}
\usepackage{framed}
\usepackage{makecell}

\def\BibTeX{{\rmB\kern-.05em{\sci\kern-.025emb}\kern-.08emT\kern-.1667em\lower.7ex\hbox{E}\kern-.125emX}}

\usepackage[colorlinks,
linkcolor=blue,
anchorcolor=blue,
citecolor=blue]{hyperref}
\usepackage{cleveref}
\usepackage{ragged2e} 
\usepackage{booktabs,makecell, multirow, tabularx}
\usepackage{pifont}
\usepackage{amsmath}
\usepackage{multicol}

\newtheorem{theorem}{Theorem}

\newtheorem{definition}[theorem]{Definition}

\usepackage{enumitem}
\usepackage{listings}

\newcommand{\sample}{\xleftarrow{\$}} 

\newcommand{\etal}{et al.}
\usepackage{pifont} 
\newcommand{\cmark}{\ding{51}} 
\newcommand{\xmark}{\ding{55}} 

\IEEEoverridecommandlockouts

\begin{document}

	\title{PriME-Deal: Privacy-Preserving Bilateral Data Trading with 
		Efficient Matchmaking and Auditable Fair Exchange on Blockchain}
	
	\author{
			Jie Zhang,~\IEEEmembership{Student Member,~IEEE},
			Xiaohong Li,~\IEEEmembership{Member,~IEEE},
			Shanshan Xu, 
			Hanwei Wu,\\
			Ruitao Feng,~\IEEEmembership{Member,~IEEE} and Guangdong Bai,~\IEEEmembership{Member,~IEEE}
			
			\thanks{This work is supported in part by the National Key Research and Development Program of China under Grant 2023YFB3107103, in part by the National Natural Science Foundation of China under Grant 62262073, 62332005, in part by the Beijing–Tianjin–Hebei Natural Science Foundation Joint Cooperation Program under Grant 25JJJJC0034.}
			\thanks{ Jie Zhang, Xiaohong Li and Hanwei Wu are with the College of Intelligence and Computing, Tianjin University, Tianjin, China. (e-mail: \{jackzhang, xiaohongli, wuhanwei\}@tju.edu.cn).}
			\thanks{ Shanshan Xu is with the School of Geographic Sciences, East China Normal University, Shanghai, China. (e-mail: s.xu.ecnu@gmail.com). }
			\thanks{ Hanwei Wu is also with the Information Security, Hainan University, China.}
			\thanks{ Ruitao Feng is with the Faculty of Science and Engineering, Southern Cross University, Australia (e-mail: ruitao.feng@scu.edu.au).}
			\thanks{ Guangdong Bai is with the Department of Computer Science, City University of Hong Kong, Hong Kong, China. (e-mail: baiguangdong@gmail.com).}
			\thanks{ Jie Zhang and Shanshan Xu are contributed equally to this work.}
			\thanks{ Hanwei Wu, Ruitao Feng and Guangdong Bai are the corresponding authors.}
		}
	
	\maketitle
	
	\begin{abstract}
		Bilateral attribute-based access control for data trading must hide policies,
		provide cryptographic fairness, and avoid trusted third parties.
		Existing solutions either leak policy information, incur super-linear costs,
		or rely on trusted dispute resolution.
		We present PriME-Deal, a non-interactive protocol that simultaneously
		achieves \emph{policy-hiding bilateral matching}, \emph{efficient threshold
			access control}, and \emph{auditable fair exchange} on public blockchains.
		The seller embeds a secret token under the buyer policy into an oblivious
		key‑value store with pseudorandom masking; the buyer reconstructs the token
		locally via tag‑based probing, eliminating combinatorial enumeration, and
		proves correctness in zero-knowledge.
		Fair exchange is enforced through a collateralized on‑chain reveal with a
		cryptographic audit that penalizes misbehaviour without trusted parties.
		
		We prove security in the Universal Composability framework under standard
		assumptions.
		Compared with the state-of-the-art threshold fuzzy IB‑ME scheme, the seller’s
		publishing time is reduced by two orders of magnitude (e.g., $8.76$\,s vs.\
		$690$\,s for a policy of 500 attributes).
		For a typical configuration of $(200,20,5)$, the buyer completes token
		reconstruction and proof generation in $8.9$\,s, with the zero-knowledge proof
		taking under $0.6$\,s and remaining constant across all parameter scales.
		The on‑chain cost is approximately $28.6$\,M gas, well within Ethereum’s block
		limit.
		PriME-Deal thus delivers the first practical privacy-preserving data
		trading protocol that combines linear seller overhead, bilateral policy hiding,
		and auditable fairness.
	\end{abstract}

	\IEEEpeerreviewmaketitle

	\section{Introduction}
	\label{sec:introduction}
	
	Decentralized data marketplaces~\cite{subramanian2023decentralized,marjanov2025sok}
	require protocols that simultaneously deliver
	\emph{bilateral attribute‑based access control}, \emph{policy privacy},
	and \emph{provable fair exchange} without trusted intermediaries~\cite{hu2024sok}.
	A seller wants to release an encrypted file only to buyers who possess
	sufficient certified attributes; a buyer will pay only for data from a
	seller meeting her own policy.
	Both parties must keep their attribute sets hidden during matching, and
	the exchange of payment for the decryption key must be atomic~\cite{DBLP:journals/tdsc/ZhangLFXHWB26}.
	Crucially, the protocol should also be \emph{non‑interactive}: the buyer
	must be able to discover and access matching data without requiring the
	seller to be online, enabling asynchronous and large‑scale marketplaces.
	
	Achieving all five goals is fundamentally difficult because of an inherent
	\emph{tension between attribute matching and data encryption}~\cite{WCXM25}.
	Matching requires computing on the two parties’ private attribute sets,
	while encryption must hide the plaintext from anyone who does not satisfy
	the matching condition.
	Prior work forces these two tasks into a single cryptographic layer, and
	this \textbf{matching--encryption coupling} inevitably sacrifices privacy,
	efficiency, fairness, or non‑interactivity.
	
	\noindent\textbf{Why existing paradigms fail.}
	(1) Matchmaking encryption (ME)~\cite{AtenieseFNV21} embeds the access policy
	directly into the ciphertext and the decryption key.
	Threshold variants~\cite{WCXM25,Wu2023Fuzzy,YCQNH26} fuse secret sharing~\cite{Shamir79}
	with pairing‑based IBE~\cite{SW05}, producing per‑attribute ciphertext
	components that leak policy information and cause \emph{quadratic}
	computational cost.
	For a policy of a few hundred attributes the encryption time becomes
	prohibitive.
	(2) Private set intersection (PSI)~\cite{KKRT16,PSWW18}, especially OKVS‑based
	designs~\cite{GPR21,RS21,RR22,BPSY23}, computes threshold intersections
	with linear cost and hides non‑intersecting elements.
	However, PSI is \emph{inherently interactive} and delivers no encrypted
	payload, nor does it authenticate the matching result~\cite{BGMP25}.
	Layering encryption on top of PSI leaks which attributes derived the key,
	undoing privacy~\cite{chielle2025recurrent}.
	(3) Blockchain‑based fair exchange~\cite{SLS18,LSB20,PKHLKO25} enforces atomic
	delivery but assumes the buyer already knows the file hash.
	It does not support attribute‑based discovery and cannot hide the buyer’s
	attributes~\cite{zhang2026qae}.
	
	\noindent\textbf{PriME‑Deal: non‑interactive decoupling.}
	We resolve the deadlock by introducing a \emph{secret token}~$\tau$ that
	cleanly separates attribute matching from file decryption, and by making the
	entire matching phase \emph{non‑interactive}.
	The seller secret‑shares $\tau$ under the buyer policy, masks each share,
	and publishes them in an oblivious key‑value store (OKVS); any buyer can
	later probe this single public structure offline and recover $\tau$ iff she
	meets the threshold---no interaction with the seller is required.
	The token $\tau$ does not decrypt the file; instead, the file key is split,
	and a second component is released only upon fair exchange.
	This architecture simultaneously delivers:
	(i)~\emph{policy‑hiding matching} without interaction,
	(ii)~\emph{linear seller and buyer costs} (no quadratic blow‑up, buyer
	work independent of her attribute set size), and
	(iii)~\emph{auditable fair exchange} via on‑chain dispute resolution
	without trusted parties.
	A zero‑knowledge proof (Groth16~\cite{Groth16}) attests to correct token
	recovery; its cost is \emph{constant} across all attribute scales.
	
	We prove security in the Universal Composability framework under standard
	bilinear and computational assumptions (DBDH, co‑CDH).
	A prototype evaluation shows that PriME-Deal reduces the seller’s
	cost by \emph{two orders of magnitude} compared with the
	state‑of‑the‑art threshold fuzzy IB‑ME~\cite{Wu2023Fuzzy}
	($8.76$\,s vs.\ $690$\,s for $500$ attributes), while keeping the buyer’s
	online time under $9$\,s for typical parameters and on‑chain verification
	gas within Ethereum’s block limit. The complete implementation for PriME is available at: \url{https://anonymous.4open.science/r/PriME-Deal/}.
	
	\noindent\textbf{Contributions.}
	\begin{enumerate}[leftmargin=*]
		\item \textbf{Conceptual.}
		We identify the \emph{matching--encryption coupling} as the central
		barrier to private, fair, and non‑interactive data trading, and
		introduce a token‑based decoupling principle that enables all five goals.
		\item \textbf{Construction.}
		We design PriME-Deal, the first non‑interactive protocol
		that simultaneously hides policies, achieves linear‑time workloads,
		and guarantees auditable fair exchange.
		\item \textbf{Security proof.}
		We provide a UC security analysis with concrete reductions to DBDH
		and co‑CDH.
		\item \textbf{Implementation \& evaluation.}
		We benchmark the full protocol, demonstrating $100\times$ seller
		speedup over the closest prior art and confirming practical buyer
		latency independent of attribute set size.
	\end{enumerate}

	\section{Related Work}
	\label{sec:related-work}
	
	Ateniese~\etal~\cite{AtenieseFNV21} introduced matchmaking encryption (ME) at
	CRYPTO~2019, enabling bilateral identity‑based access control where decryption
	succeeds only for the designated sender–receiver pair.
	Subsequent identity‑based constructions~\cite{francati2021identity,CLWW22,BL23,Sun23PSME}
	remain limited to exact matching.
	To support threshold policies, Wu~\etal~\cite{Wu2023Fuzzy} proposed Fuzzy
	IB‑ME, where identities are attribute sets and decryption requires overlap on
	both sides above a threshold.
	Further extensions add anonymous credentials~\cite{MXNHD23}, unbounded
	universes~\cite{YDY0W24}, lightweight revocability for IoT~\cite{WWCCW26},
	and revocable multi‑receiver IB‑ME~\cite{WFZYC24}.
	Despite these advances, all existing ME schemes embed the access policy
	directly into pairing‑based decryption, making decryption cost at least
	linear in the policy size and leaking partial attribute information.
	PriME‑Deal departs from this paradigm: we encode PRF‑masked Shamir shares of a
	token $\tau$ into a public OKVS; the buyer reconstructs $\tau$ off‑line using
	her attributes, after which file decryption requires only \emph{one}
	constant‑time pairing.
	This decoupling of matching from decryption is the key architectural
	innovation that makes the scheme practical for unbalanced policies.
	
	On the orthogonal front, private set intersection (PSI) enables two parties
	to compute the intersection of their private sets without leaking additional
	information.
	The efficiency landscape was transformed by the linear OKVS for
	PSI~\cite{GPR21} (CRYPTO~2021), which encodes key‑value pairs into a compact
	vector that hides the keys.
	Bienstock~\etal~\cite{BPSY23} pushed the encoding rate to $0.97$, and
	Wu~\etal~\cite{WYY25} combined OKVS with OT extension for a 22--41\%
	speedup over Kolesnikov~\etal~\cite{KKRT16}.
	Although threshold and fuzzy PSI~\cite{CDGOSS21,YCW24,PSTKZ25} relax exact
	matching, all existing PSI variants are inherently \emph{interactive}:
	both parties must engage in multi‑round protocols, and the output is merely a
	set relation without any mechanism to bind the result to a decryptable
	ciphertext.
	In the multi‑party setting, fair PSI via smart contracts~\cite{AS24} proves
	that complete fairness is impossible without a trusted third party,
	motivating the use of economic collateral as in our design.
	PriME‑Deal uses OKVS in a \emph{non‑interactive} manner: the seller encodes
	masked shares once, and any buyer later probes the vector offline with her own
	attributes.
	The double masking guarantees \emph{value‑hiding}: an adversary without the
	attribute decryption key cannot distinguish a decoded value from random.
	Crucially, the reconstructed token $\tau$ is directly bound to a decryptable
	file key, coupling set intersection with encrypted data delivery.
	
	Regarding fair exchange, FairSwap~\cite{SLS18} (CCS~2018) pioneered smart
	contracts as external judges, with buyers submitting Merkle proofs to dispute
	misbehavior.
	OptiSwap~\cite{LSB20} improved the optimistic path and added protection
	against grieving attacks.
	To further guarantee privacy, zkMarket~\cite{PKHLKO25} integrates
	commit‑and‑prove SNARKs with smart contracts, while PVSS~\cite{OLMZK25} and
	SwiftGuard~\cite{ZLZHBF25} use zk‑SNARKs for public verifiability.
	Piper~\etal~\cite{PWH25} combines ZKPs with self‑sovereign identity and
	attribute‑based access control for DeFi, though it does not bind matching
	outcomes to encrypted payloads.
	All these fair exchange protocols, even those enhanced with zero‑knowledge
	proofs, require the buyer to know the file hash before purchasing.
	In a privacy‑preserving data marketplace, the buyer must first \emph{discover}
	compatible data through attribute‑based matching.
	PriME‑Deal bridges this gap by integrating OKVS‑based matching with an
	on‑chain audit mechanism, enabling a genuine \emph{discovery‑then‑purchase}
	workflow, and couples a zk‑SNARK compliance proof directly with the fair
	exchange contract.
	
	Table~\ref{tab:feature-compare} provides a qualitative comparison of
	PriME‑Deal with the most representative related work across six key
	dimensions.
	
	\begin{table}[htbp]
		\centering
		\caption{Qualitative comparison with representative schemes}
		\label{tab:feature-compare}
		\renewcommand{\arraystretch}{1.3}
		\setlength{\tabcolsep}{1.8pt}
		\resizebox{\columnwidth}{!}{%
			\begin{tabular}{lcccccc}
				\toprule
				\textbf{Scheme} &
				\makecell[c]{\textbf{Bilateral}\\ \textbf{Policy}} &
				\makecell[c]{\textbf{Threshold}\\ \textbf{Matching}} &
				\makecell[c]{\textbf{Policy}\\ \textbf{Privacy}} &
				\makecell[c]{\textbf{Non‑inter.}\\ \textbf{Matching}} &
				\makecell[c]{\textbf{Fair}\\ \textbf{Exchange}} &
				\makecell[c]{\textbf{ZK}\\ \textbf{Compl.}}\\
				\midrule
				IB‑ME~\cite{AtenieseFNV21,CLWW22}
				& \cmark & \xmark & \xmark & \cmark & \xmark & \xmark\\
				Fuzzy IB‑ME~\cite{Wu2023Fuzzy}
				& \cmark & \cmark & \cmark\textsuperscript{*} & \cmark & \xmark & \xmark \\
				Fuzzy PSI~\cite{PSTKZ25}
				& \xmark & \cmark & \cmark & \xmark & \xmark & \xmark \\
				FairSwap~\cite{SLS18} / OptiSwap~\cite{LSB20}
				& \xmark & \xmark & \xmark & \cmark\textsuperscript{\#} & \cmark & \xmark \\
				zkMarket~\cite{PKHLKO25}
				& \xmark & \xmark & \xmark & \cmark\textsuperscript{\#} & \cmark & \cmark \\
				Piper~\cite{PWH25}
				& \xmark & \cmark & \cmark & \cmark & \xmark & \cmark \\
				\midrule
				\textbf{PriME‑Deal (Ours)}
				& \cmark & \cmark & \cmark & \cmark & \cmark & \cmark \\
				\bottomrule
			\end{tabular}%
		}
		\begin{flushleft}
			\footnotesize
			\cmark: supported; \xmark: not supported.\\
			\textsuperscript{*}Partial: a matching party learns which attributes belong to the policy.\\
			\textsuperscript{\#}These schemes require the buyer to know the file hash in advance,
			so matching is trivial (exact identifier lookup) and not attribute‑based.
		\end{flushleft}
	\end{table}

	\section{Preliminaries}
	\label{sec:prelim}
	
	We denote by $\lambda$ the security parameter; throughout this work we target
	$\lambda = 128$ bits of symmetric security.  Let $\mathbb{F}_p$ be a finite field
	of prime order $p$ with $p \approx 2^{2\lambda}$, and let $x \sample \mathcal{D}$
	denote uniform sampling from distribution $\mathcal{D}$.  All algorithms are
	probabilistic polynomial-time (PPT) unless otherwise stated.  For a set
	$\mathcal{X}$, $|\mathcal{X}|$ denotes its cardinality.  We write $[n]$ for
	$\{1,\dots,n\}$. See Appendix~\ref{app:table} for an explanation of the iconic symbols used in PriME-Deal.
	
	\subsection{Bilinear Groups and Complexity Assumptions}
	\label{sec:bilinear}
	
	Let $\mathbb{G}, \mathbb{G}_1, \mathbb{G}_T$ be cyclic groups of prime order $p$ and
	$e: \mathbb{G} \times \mathbb{G}_1 \to \mathbb{G}_T$ a non‑degenerate, efficiently
	computable bilinear pairing.  Let $g$ and $g_1$ be generators of $\mathbb{G}$ and
	$\mathbb{G}_1$, respectively.  We instantiate the bilinear group with the
	\textbf{BN254} curve (Type‑3 asymmetric pairing~\cite{AKLGL11}), which provides
	approximately $128$ bits of security and is natively supported by Groth16
	zk‑SNARK libraries.
	
	\begin{definition}[Computational co‑Diffie–Hellman (co‑CDH)]
		Given $g, g_1^a \in \mathbb{G}$ and $g_1 \in \mathbb{G}_1$, no PPT adversary
		can compute $g_1^a \in \mathbb{G}_1$ with more than negligible probability.
	\end{definition}
	
	\begin{definition}[Decisional Bilinear Diffie–Hellman (DBDH)]
		Given $(g^a, g^b, g^c \in \mathbb{G},\; g_1^a, g_1^b, g_1^c \in \mathbb{G}_1)$
		and $Z \in \mathbb{G}_T$, where $a,b,c \sample \mathbb{Z}_p^*$,
		no PPT adversary can distinguish $Z = e(g,g_1)^{abc}$ from a random element
		of $\mathbb{G}_T$ with more than negligible advantage.
	\end{definition}
	
	\begin{definition}[Computational Bilinear Diffie–Hellman (CBDH)]
		Given $g, g^a, g^b, g^c \in \mathbb{G}$ and $g_1 \in \mathbb{G}_1$,
		no PPT adversary can compute $e(g,g_1)^{abc}$ with more than negligible
		probability.
	\end{definition}
	
	\subsection{Hash Functions and Pseudorandom Function}
	\label{sec:hash-prf}
	We use four hash functions, all modeled as random oracles:
	$H_1: \{0,1\}^* \to \mathbb{G}_1$ (full‑domain hash into $\mathbb{G}_1$),
	$H_2: \mathbb{G}_T \to \{0,1\}^{\lambda}$,
	$H_3: \{0,1\}^* \to \mathbb{F}_p$, and
	$H_4: \{0,1\}^* \to \{0,1\}^{\lambda}$.
	
	A pseudorandom function $\mathsf{PRF}: \{0,1\}^\lambda \times \{0,1\}^* \to \mathbb{F}_p$
	is employed such that $\mathsf{PRF}(k,\cdot)$ is computationally indistinguishable
	from a random function when $k \sample \{0,1\}^\lambda$~\cite{KKRT16}. In the protocol, we use two independent PRF keys, denoted $\mathsf{s}_\nu$ and $\mathsf{s}_r$, to provide pseudorandom masking and deterministic encryption randomness, respectively. We instantiate it with HMAC-SHA256, truncating the output modulo $p$,
	and write $\mathsf{PRF}(k, x)$ for the evaluation on input $x$ under key $k$.
	
	\subsection{Authenticated Symmetric Encryption}
	\label{sec:aes}
	We use AES-128-GCM as the authenticated encryption scheme for all file
	encryptions.  The scheme operates with 128-bit keys and consists of:
	
	\begin{itemize}[leftmargin=10px]
		\item $\mathsf{AES.Enc}(K,M) \to C$: samples a fresh 96-bit nonce $N \sample \{0,1\}^{96}$, computes the AES-GCM encryption $(C', T) \gets \text{AES-GCM.Enc}(K, N, M, \varepsilon)$ where $C'$ is the ciphertext, $T$ is the 128-bit authentication tag, $\varepsilon$ is the additional authenticated data and outputs $C \gets N \| C' \| T$.
		\item $\mathsf{AES.Dec}(K,C) \to M$ or $\bot$: parses $C$ as $N \| C' \| T$, computes $(M, \text{valid}) \gets \text{AES-GCM.Dec}(K, N, C', \varepsilon, T)$, and outputs $M$ if $\text{valid} = \text{true}$, otherwise outputs $\bot$.
	\end{itemize}
	
	\noindent
	AES-128-GCM provides IND-CPA confidentiality and INT-CTXT ciphertext integrity, and its 128-bit key size matches our targeted security level
	$\lambda = 128$.  Throughout the protocol, file contents are encrypted with
	$\mathsf{AES.Enc}$ and decrypted with $\mathsf{AES.Dec}$, $\varepsilon$ is the empty string.
	
	\subsection{Linear Oblivious Key-Value Store (OKVS)}
	\label{sec:okvs}
	
	A linear OKVS over a field $\mathbb{F}$ is defined by a deterministic row
	function $\mathsf{row}: \{0,1\}^* \to \mathbb{F}^m$ and two algorithms
	\cite{BPSY23,XLW25}:
	\begin{itemize}[leftmargin=10px]
		\item $\mathsf{Encode}(\{(k_i, v_i)\}_{i=1}^n) \to D$: given $n$
		key-value pairs, outputs a vector $D \in \mathbb{F}^m$ such that
		$\langle \mathsf{row}(k_i), D \rangle = v_i$ for all $i$.
		\item $\mathsf{Decode}(D, k) = \langle \mathsf{row}(k), D \rangle$.
	\end{itemize}
	The parameter $m$ is chosen as $m \approx 1.3\,n$ to guarantee that encoding
	succeeds with overwhelming probability.  \emph{Obliviousness} guarantees that
	$D$ leaks no information about the set of encoded keys beyond its dimension $m$:
	for any two key sets $\mathcal{X}_0, \mathcal{X}_1$ of equal size $n$ and
	uniformly random values $\{v_x\}$, the distributions of
	$\mathsf{Encode}(\{(x, v_x)\}_{x \in \mathcal{X}_0})$ and
	$\mathsf{Encode}(\{(x, v_x)\}_{x \in \mathcal{X}_1})$ are computationally
	indistinguishable.
	
	In PriME-Deal, we encode PRF-masked Shamir shares into the OKVS,
	which guarantees \emph{value-hiding}: even for a key $x \notin \mathcal{X}$,
	$\mathsf{Decode}(D, x)$ is computationally indistinguishable from a random
	field element.  Formally, for any PPT adversary $\mathcal{A}$ that is not given
	the PRF keys, the advantage in distinguishing $\mathsf{Decode}(D, x)$ from a
	uniform element is negligible.
	
	\subsection{$(t,n)$-Shamir Secret Sharing}
	\label{sec:shamir}
	A $(t,n)$-Shamir secret sharing scheme \cite{Shamir79} splits a secret 
	$\tau \in \mathbb{F}_p$ into $n$ shares by choosing a random polynomial 
	$f(x) \in \mathbb{F}_p[x]$ of degree $t-1$ with $f(0) = \tau$ and defining 
	the $i$-th share as $s_i = f(x_i)$, where $\{x_i\}_{i=1}^n$ are distinct, 
	non-zero elements of $\mathbb{F}_p$.
	Any subset of $t$ shares $\{s_i\}_{i \in S}$ (with $S \subset \{1,\dots,n\}$, $|S| = t$) 
	uniquely determines $\tau$ via Lagrange interpolation:
	\[
	\tau = \sum_{i \in S} s_i \cdot \ell_i(0), \quad 
	\ell_i(x) = \prod_{j \in S,\, j \neq i} \frac{x - x_j}{x_i - x_j}.
	\]
	Fewer than $t$ shares reveal no information about $\tau$.
	
	\subsection{Cuckoo Filter}
	\label{sec:cuckoo}
	
	A Cuckoo filter \cite{PSWW18,FanAndersen13cuckoo} is a probabilistic data structure
	supporting approximate set membership queries with tunable false-positive rate
	$\epsilon$.  It provides two operations:
	$\mathsf{CF}.\mathsf{insert}(x)$ inserts an element $x$,
	and $\mathsf{CF}.\mathsf{lookup}(x)$ returns \texttt{true} if $x$ may be in
	the set, \texttt{false} if it is definitely not.
	
	\subsection{Zero-Knowledge Succinct Non-Interactive Argument of Knowledge}
	\label{sec:zk}
	
	A zk-SNARK~\cite{Groth16} for an NP relation $\mathcal{R}$ is a triple
	$(\mathsf{Setup}, \mathsf{Prove}, \mathsf{Verify})$ with:
	$\mathsf{Setup}(1^\lambda) \to \mathsf{crs}$ generates a common reference string,
	$\mathsf{Prove}(\mathsf{crs}, x, w) \to \pi$ produces a proof $\pi$ for
	$(x,w) \in \mathcal{R}$, and
	$\mathsf{Verify}(\mathsf{crs}, x, \pi) \to \{0,1\}$ checks validity.
	It must satisfy perfect completeness, computational knowledge soundness, and
	computational zero-knowledge. 
	In our construction, we instantiate $\mathsf{Com}(m; \omega) := H_4(m \| \omega)$ where $\omega \sample \{0,1\}^\lambda$, and write $\mathsf{com}_{\mathcal{X}} \gets \mathsf{Com}(\mathcal{X}; \omega)$ for the
	commitment to a set $\mathcal{X}$.

	\begin{figure*}[htbp]
		\centering
		\includegraphics[width=\linewidth]{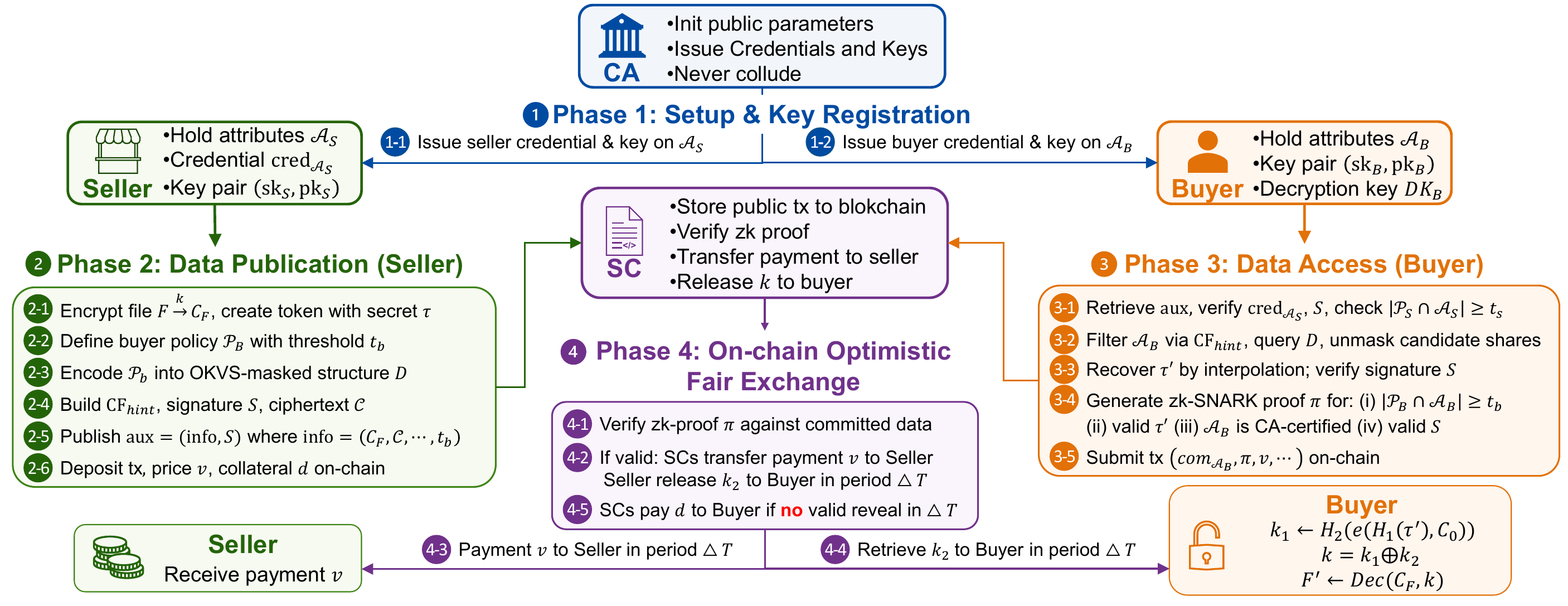}
		\caption{The system-overview of PriME-Deal.}
		\label{fig:system-overview}
	\end{figure*}

	\section{System Overview}
	\label{sec:system}
	
	As illustrated in Fig.~\ref{fig:system-overview}, the PriME-Deal ecosystem consists of four core entities and a blockchain-based public ledger.
	Data transmission between entities can be carried out via any standard secure channel (e.g., TLS or a peer-to-peer protocol); the protocol itself is agnostic to the specific transport mechanism.
	Only succinct cryptographic commitments and proofs are recorded on-chain.
	
	\begin{enumerate}[leftmargin=*]
		\item \textbf{Certificate Authority (CA).} A fully trusted entity that
		runs $\mathsf{Setup}$ to generate global parameters, registers users, and
		issues long‑term attribute‑based credentials.  The CA is assumed honest
		and never colludes.
		
		\item \textbf{Seller ($\mathcal{S}$).} Holds certified attributes
		$\mathcal{A}_S$ and sells decryption rights for a file $F$.  It defines a
		buyer policy $\mathcal{P}_B$ with threshold $t_b$, encodes it into a
		public OKVS, encrypts $F$, and deposits the key hash together with the
		price and an over‑collateral into the smart contract.
		
		\item \textbf{Buyer ($\mathcal{B}$).} Holds certified attributes
		$\mathcal{A}_B$ and specifies a seller policy $\mathcal{P}_S$ with
		threshold $t_s$.  It locally reconstructs the token from the OKVS and
		generates a ZKP attesting that its own attributes satisfy the seller's
		buyer policy (i.e., $|\mathcal{P}_B \cap \mathcal{A}_B| \ge t_b$).
		Before generating the proof, the buyer also locally verifies that the
		seller's attributes meet its own seller policy
		($|\mathcal{P}_S \cap \mathcal{A}_S| \ge t_s$); if not, it aborts.
		This design enforces \textbf{bilateral matching} while only the
		buyer‑side check is submitted on‑chain via the ZKP, keeping the
		seller‑side policy and the buyer's attributes private.
		
		\item \textbf{Smart Contract ($\mathcal{SC}$).} A self‑executing program
		that holds the encrypted key hash in escrow, verifies the buyer's ZKP,
		and manages an optimistic fair‑exchange protocol.  After a valid proof
		is submitted, the contract freezes the payment and opens a challenge
		period during which the seller must reveal the partial decryption key.
		If the buyer detects a decryption failure, it can trigger an on‑chain
		audit procedure that either confirms the key's correctness or penalises
		a cheating seller by transferring the collateral to the buyer.  The
		contract thus acts as a decentralized mediator requiring no trust.
	\end{enumerate}

	%
	
	\subsection{Protocol Overview}
	\label{sec:protocol-overview}
	
	To realize the above system, PriME‑Deal comprises five main algorithms.
	We briefly state their interfaces; full specifications appear in Sections~\ref{sec:setup}--\ref{sec:exchange}.
	
	\begin{enumerate}[leftmargin=*]
		\item $\mathsf{Setup}(1^\lambda) \to (\mathsf{pp}, \mathsf{msk})$: generates bilinear groups, hash functions, a linear OKVS instance, the ZKP common reference string $\mathsf{crs}$, and key pair $(\mathsf{msk},\mathsf{mpk})$.
		\item $\mathsf{AttrKeyGen}(\mathsf{msk}, \mathsf{uid}, \mathcal{A}) \to (\mathsf{IC}_\mathsf{uid}, \mathsf{DK}_\mathsf{uid})$: the CA generates an identity credential to a registered user. For a seller, the credential enables token signature generation; for a buyer, the CA additionally computes a decryption key serves as a certified witness in the ZKP.
		\item $\mathsf{Publish}(\mathsf{pp}, \mathsf{EK}_S, \mathcal{A}_S, \mathcal{P}_B, t_b, F, v) \to \mathsf{aux}$: the seller generates a fresh token $\tau$, secret‑shares it under $\mathcal{P}_B$, masks the shares with a PRF, and encodes them into a public OKVS $D$.
		It also encrypts $F$ under a key derived from $\tau$, publishes the auxiliary data $\mathsf{aux}$, and deposits the encrypted key hash together with the price $v$ and a collateral $d$ into $\mathcal{SC}$.
		\item $\mathsf{Access}(\mathsf{pp}, \mathsf{DK}_B, \mathcal{P}_S, t_s, \mathcal{A}_B, \mathsf{aux}) \to (\pi, \mathsf{tx})$:
		The buyer first locally checks that the seller's attributes satisfy its own policy,
		i.e., $|\mathcal{P}_S \cap \mathcal{A}_S| \ge t_s$, and aborts if not.
		It then probes the OKVS, reconstructs the token $\tau$ using its own attributes
		$\mathcal{A}_B$, and generates a ZKP $\pi$ attesting that
		(i)~$\tau$ was correctly recovered from the OKVS using a subset of~$\mathcal{A}_B$,
		which implies $|\mathcal{P}_B \cap \mathcal{A}_B| \ge t_b$, and
		(ii)~the credential and seller's signature on $\tau$ are valid.
		The proof and payment are bundled into a blockchain transaction $\mathsf{tx}$.
		\item \textbf{On‑Chain Optimistic Fair Exchange.}
		The smart contract $\mathcal{SC}$ verifies $\pi$; if valid, it freezes the
		buyer's payment $v$ and opens a challenge window.  The seller then reveals
		the partial key $k_2$ on‑chain.  If the buyer's subsequent decryption
		succeeds, the payment is released to the seller.  If decryption fails, the
		buyer may invoke an on‑chain audit (Section~\ref{sec:audit-evidence}) that
		checks the consistency of $k_2$ with the encrypted file via Merkle proofs.
		A successful audit proves seller misbehaviour, causing the collateral $d$
		to be transferred to the buyer and the payment $v$ to be refunded.  This
		mechanism guarantees fairness without assuming rational behaviour from the
		seller.
	\end{enumerate}
	
	In our PriME‑Deal, the OKVS vector $D$ is public, yet due to PRF masking and obliviousness, an adversary learns nothing about $\mathcal{P}_B$ beyond what can be inferred from its length. The token $\tau$ acts as a bridge: it is hidden by the OKVS, authenticated by a signature, and finally used in a single pairing to recover the file key. Only a constant number of pairings (two for signature verification and key decryption) are performed by the buyer, independent of the policy sizes.

	\subsection{Threat Model}
	\label{sec:threat-model}
	
	We consider probabilistic polynomial-time (PPT) adversaries.
	The certificate authority (CA) is fully trusted and never corrupted.
	The adversary can statically corrupt any subset of sellers and buyers, learning
	their long‑term secrets ($\mathsf{sk}_\mathsf{uid}$, $\mathsf{DK}$, $\mathsf{cred}$).
	All public communication, such as blockchain transactions, auxiliary data $\mathsf{aux}$,
	and on‑chain events, is visible to the adversary.
	Corrupted parties may deviate arbitrarily from the protocol (malicious),
	whereas uncorrupted parties follow the protocol but may try to infer private
	information from their views (honest‑but‑curious).
	External observers that only monitor public data are also captured by the model.
	\textbf{Note:} network attacks such as DDoS, physical side‑channel attacks
	are out of scope.
	
	\subsection{Security Model and Formal Definitions}
	\label{sec:security-model}
	
	We provide two complementary formalizations: an ideal functionality in the
	Universal Composability (UC) framework that captures the overall protocol
	guarantees, and game‑based definitions that directly align with our concrete
	assumptions and allow precise reductions.
	
	\subsubsection{UC Ideal Functionality $\mathcal{F}_{\mathsf{PPFE}}^{\mathsf{weak}}$}
	\label{sec:uc-func}
	
	Public blockchains cannot hide data from observers, and a smart contract
	cannot directly force an off‑chain party to reveal a secret.
	Consequently, a fully atomic ideal functionality that instantly delivers the
	decryption key upon a valid proof is impossible without a trusted third party.
	Following the approach of FairSwap~\cite{SLS18} and OptiSwap~\cite{LSB20},
	we capture this inherent limitation by a \emph{weak} functionality in which
	the seller must actively perform a key‑reveal step, and fairness is enforced
	through an on‑chain challenge‑response mechanism with a penalty deposit.
	The gap between our weak functionality and a fully atomic one is discussed
	at the end of this subsection.
	
	\textbf{Functionality $\mathcal{F}_{\mathsf{PPFE}}^{\mathsf{weak}}$:~}
	
	\noindent \textbf{State.}
	A table $\mathsf{Listings}[\mathsf{sid}]$ stores
	$(\mathsf{seller}, \mathcal{A}_S, \mathcal{P}_B, t_b, v, d, K, h_K, \tau,
	\mathsf{status}, \mathsf{timer}, \mathsf{pub})$,
	where $\mathsf{pub}$ collects the public values needed to verify the
	buyer's ZKP, namely
	$(h_D, \mathsf{s}_\nu, S_0, C_0, \mathsf{mpk}, \mathsf{pk}_S, \mathsf{uid}_S)$.
	Initially $\mathsf{status}=\bot$.
	
	\begin{itemize}[leftmargin=10px]
		\item \textbf{Registration.}
		On input $(\mathsf{Register}, \mathsf{uid}, \mathcal{A})$ from the CA,
		store $\mathcal{A}$ under identity $\mathsf{uid}$.
		
		\item \textbf{Publishing.}
		On input $(\mathsf{Publish}, \mathsf{sid},
		\mathcal{A}_S, \mathcal{P}_B, t_b, F, v, d, \mathsf{pub})$ from a
		registered seller:
		generate a random file key $K \sample \{0,1\}^\lambda$ and a fresh token
		$\tau \sample \mathbb{F}_p$, set $h_K = H_4(K)$,
		store $(\mathsf{seller}, \mathcal{A}_S, \mathcal{P}_B, t_b, v, d, K, h_K,
		\tau, \text{active}, \bot, \mathsf{pub})$, and send
		$(\mathsf{sid}, \mathsf{seller}, |\mathcal{A}_S|, |\mathcal{P}_B|, t_b, v, d)$
		to $\mathcal{S}$.
		
		\item \textbf{Probing.}
		On input $(\mathsf{Probe}, \mathsf{sid}, x)$ from $\mathcal{S}$:
		if no listing $\mathsf{sid}$ exists, return $\bot$.
		If $x \in \mathcal{P}_B$, return $\top$; otherwise return a
		random element of $\mathbb{F}_p$.  This models the OKVS+PRF
		value‑hiding guarantee.
		
		\item \textbf{Buy (initiating challenge).}
		On input $(\mathsf{Buy}, \mathsf{sid}, \mathcal{P}_S, t_s,
		\mathcal{A}_B, \pi, v)$ from a registered buyer:
		\begin{itemize}
			\item If the listing is not active, reject.
			\item Let $\mathsf{pub}$ be the stored public data.
			Verify that $\pi$ is a valid proof for $\mathcal{R}_{\mathsf{buy}}$
			with the public statement $(\mathsf{sid}$, $\mathsf{com}_{\mathcal{A}_B}$,
			$\mathsf{pub}.h_D$, $\mathsf{pub}.\mathsf{s}_\nu$, $t_b$, $\mathsf{pub}.S_0$,
			$\mathsf{pub}.C_0$, $\mathsf{mpk}$, $\mathsf{pub}.\mathsf{pk}_S$,
			$\mathsf{pub}.\mathsf{uid}_S)$.
			\item If $|\mathcal{P}_B \cap \mathcal{A}_B| < t_b$ or the proof is
			invalid, reject.
			\item \textbf{Seller‑side policy.}
			$\mathcal{F}$ does \emph{not} check $|\mathcal{P}_S \cap \mathcal{A}_S|
			\ge t_s$.  Instead, the buyer is expected to perform this check
			locally before calling $\mathsf{Buy}$; a malicious buyer that skips
			the check may lose its payment but cannot break fairness.
			\item Otherwise, freeze $v$ from the buyer, set
			$\mathsf{status} \gets \text{challenged}$,
			$\mathsf{timer} \gets \Delta T$, and notify $\mathcal{S}$ of the
			challenge with leakage $(\mathsf{sid}, \mathsf{com}_{\mathcal{A}_B},
			\text{timestamp})$.
		\end{itemize}
		
		\item \textbf{Key Reveal.}
		On input $(\mathsf{KeyReveal}, \mathsf{sid}, K')$ from the seller:
		\begin{itemize}
			\item If $\mathsf{status} \neq \text{challenged}$ or $\mathsf{timer}$
			has expired, reject.
			\item If $H_4(K') \neq h_K$, reject.
			\item Transfer $v$ to the seller, return $d$ to the seller,
			send $K'$ to the buyer, set $\mathsf{status} \gets \text{completed}$.
		\end{itemize}
		
		\item \textbf{Timeout.}
		When $\mathsf{status} = \text{challenged}$ and $\mathsf{timer}$
		expires (can be triggered by any party):
		return $v$ to the buyer, transfer $d$ to the buyer, set
		$\mathsf{status} \gets \text{disputed}$.
	\end{itemize}
	
	\noindent \textbf{Comparison with a fully atomic functionality.}
	A fully atomic $\mathcal{F}_{\mathsf{PPFE}}$ would deliver the decryption key
	$K$ immediately upon a valid $\mathsf{Buy}$ request, without requiring the
	seller to perform an explicit reveal step.  Such a functionality cannot be
	realised on public blockchains without a trusted party because the contract
	cannot extract $K$ from the seller.  Our weak functionality faithfully
	reflects this limitation and admits a concrete protocol, while still
	providing strong fairness guarantees through the collateral and audit
	mechanism.

	\subsubsection{Game‑Based Definitions}
	\label{sec:games}
	
	For concrete reductions we define two games.
	
	\textbf{(1) Policy Privacy Game} $\mathsf{Priv}_{\mathcal{A}}^{\mathsf{policy}}(\lambda)$:
	\begin{itemize}[leftmargin=10px]
		\item \textbf{Init.} $\mathcal{A}$ outputs two seller attribute sets
		$\mathcal{A}_{S,0}, \mathcal{A}_{S,1}$ of equal size, two buyer policies
		$\mathcal{P}_{B,0}, \mathcal{P}_{B,1}$ of equal size, and thresholds
		$t_s, t_b$.
		\item \textbf{Setup.} The challenger runs $\mathsf{Setup}$ to obtain
		$\mathsf{pp}$ and gives it to $\mathcal{A}$.
		\item \textbf{Key Queries.} $\mathcal{A}$ may adaptively request
		decryption keys for buyer identities and signing keys for seller
		identities.  The adversary is restricted so that \emph{no} queried
		buyer identity $\mathsf{uid}_B$ (with attribute set $\mathcal{A}_B$)
		simultaneously satisfies $|\mathcal{P}_{B,0} \cap \mathcal{A}_B| \ge t_b$
		and $|\mathcal{P}_{B,1} \cap \mathcal{A}_B| \ge t_b$.  (If such a
		buyer existed, it could trivially decrypt the file for both challenge
		branches and distinguish $b$.)
		\item \textbf{Challenge.} $\mathcal{A}$ chooses two equal‑length files
		$F_0, F_1$.  The challenger picks a random bit $b$, runs $\mathsf{Publish}$
		with $(\mathcal{A}_{S,b}, \mathcal{P}_{B,b}, t_b, F_b)$, and returns the
		resulting $\mathsf{aux}$ and on‑chain deposit.
		\item \textbf{Guess.} $\mathcal{A}$ outputs $b'$ and wins if $b' = b$.
	\end{itemize}
	The scheme is \textbf{policy‑private} if
	$\Pr[\mathcal{A}\text{ wins}] \le 1/2 + \mathsf{negl}(\lambda)$.

	\textbf{(2) Fair Exchange Game} $\mathsf{Fair}_{\mathcal{A}}^{\mathsf{exch}}(\lambda)$:~
	This game provides the atomicity guarantee that the buyer obtains the file key iff the seller receives the payment, even in the presence of a malicious party.
	\begin{itemize}[leftmargin=10px]
		\item \textbf{Setup.} As above.
		\item \textbf{Corruption.} $\mathcal{A}$ outputs a corruption bit $c$
		(with the convention that $c=0$ means $\mathsf{seller}$ and $c=1$ means
		$\mathsf{buyer}$) and receives the corresponding long‑term keys.
		The other party is honest.
		\item \textbf{Execution.} The challenger plays the role of the honest party
		and faithfully simulates the blockchain smart contract (including the
		collateral and challenge period).  $\mathcal{A}$ may arbitrarily deviate
		from the prescribed algorithms. In particular:
		\begin{itemize}
			\item When the seller is corrupted, $\mathcal{A}$ may choose
			inconsistent $(k_2, h_{k_2})$ during $\mathsf{Publish}$ and later
			reveal a different $k_2'$ that satisfies the on‑chain hash but results
			in decryption failure.
			\item When the buyer is corrupted, $\mathcal{A}$ may attempt to
			forge a ZK proof without valid attributes or collude with others to
			combine keys.
		\end{itemize}
		The execution runs until the protocol terminates or $\mathcal{A}$ aborts;
		afterwards the winning condition is evaluated.
		\item \textbf{Winning Condition.} Define two events:
		\begin{align*}
			E_1 &: \text{buyer obtains decryption key } k \text{ for } F \\
			& (i.e., \text{get a key to make}\mathsf{AES.Dec}(k, C_F) \neq \bot )\\
			E_2 &: \text{seller receives payment } v
		\end{align*}
		\item \textbf{Challenge.} For an honest buyer the challenger directly verifies $E_1$ by checking
		whether it received and successfully used such a key; for a corrupted
		buyer the challenger inspects the protocol transcript and decides whether
		a key $k$ satisfying the condition was delivered to $\mathcal{A}$
		(e.g., via an on‑chain reveal). $E_2$ is evaluated by checking the
		payment state recorded by the challenger.
		\item \textbf{Win.} $\mathcal{A}$ wins if $E_1$ and $E_2$ differ in truth value, i.e., exactly
		one of them holds.
	\end{itemize}
	The scheme is \textbf{fair-exchange} if for every PPT $\mathcal{A}$,
	$\Pr[\mathcal{A}\text{ wins}] \le \mathsf{negl}(\lambda)$.

	\section{Construction of PriME‑Deal}
	\label{sec:protocol}
	We now present the formal specification of PriME‑Deal.
	The protocol comprises four stages: (i) system initialization and key
	distribution, (ii) data publication by the seller, (iii) bilateral
	verification, local token reconstruction and zero‑knowledge proof
	generation by the buyer, and (iv) on‑chain auditable fair exchange followed by
	off‑chain file decryption.
	The complete flow is depicted in Appendix.\ref{fig:prime-deal-full}.
	
	\subsection{System Setup}
	\label{sec:setup}
	$\mathsf{Setup}(1^\lambda, n_{\max})$ generates a bilinear group $(\mathbb{G}, \mathbb{G}_1, \mathbb{G}_T, p, g, g_1, e)$
	and selects random oracles $H_1:\{0,1\}^*\!\to\!\mathbb{G}_1$, $H_2:\mathbb{G}_T\!\to\!\{0,1\}^\lambda$,
	$H_3:\{0,1\}^*\!\to\!\mathbb{F}_p$, $H_4:\{0,1\}^*\!\to\!\{0,1\}^\lambda$.
	A linear OKVS over $\mathbb{F}_p$ with row function $\mathsf{row}:\{0,1\}^*\!\to\!\mathbb{F}_p^m$
	($m\!\approx\!1.3 n_{\max}$) and a PRF $\mathsf{PRF}:\{0,1\}^\lambda\!\times\!\{0,1\}^*\!\to\!\mathbb{F}_p$ are instantiated.
	The master secret key is $\mathsf{msk}=s\sample\mathbb{F}_p$ with $\mathsf{mpk}=g^s$.
	A Groth16 CRS is generated for $\mathcal{R}_{\mathsf{buy}}$ (Section~\ref{sec:buyer-circuit}) over BN254.
	Output $\mathsf{pp}=(\mathbb{G},\mathbb{G}_1,\mathbb{G}_T,p,g,g_1,\mathsf{mpk},e,H_1,H_2,H_3,H_4,\mathsf{row},\mathsf{PRF},\mathsf{crs})$.
	
	\subsection{Identity Credential and Attribute Key Generation}
	\label{sec:attr-key}
	For a user with identity $\mathsf{uid}$ and attribute set $\mathcal{A}$,
	the CA computes $\mathsf{sk}=s\cdot H_3(\mathsf{uid})$, $\mathsf{pk}=g^{\mathsf{sk}}$,
	and $\mathsf{cred}_{\mathcal{A}}=H_1(\mathsf{uid}\|\mathcal{A})^{s}$.
	\begin{itemize}
		\item \textbf{Seller} publishes $\mathsf{IC}_S=(\mathsf{uid}_S,\mathsf{pk}_S,\mathsf{cred}_{\mathcal{A}_S})$
		with $\mathcal{A}_S$ in the clear.
		\item \textbf{Buyer} additionally obtains $\mathsf{DK}_B=\{(a,D_a)\}_{a\in\mathcal{A}_B}$
		where $D_a=H_1(a)^{s}$, used only for unmasking during token reconstruction.
		The identity credential is $\mathsf{IC}_B=(\mathsf{uid}_B,\mathsf{pk}_B,\mathsf{cred}_{\mathcal{A}_B})$.
	\end{itemize}
	
	\subsection{Publish (Seller)}
	\label{sec:publish}
	
	Algorithm~\ref{alg:publish} details the seller's off‑line procedure. In seller-side, the buyer policy $\mathcal{P}_B$ is never transmitted; its Shamir shares
	are doubly masked using a PRF and the hash of an attribute‑based
	encapsulation key before being encoded into the public OKVS $D$.
	Each ciphertext $C_i \in \mathcal{C}$ is a pair $(C_i^*, \mathsf{tag}_i)$ where
	$C_i^* = g^{r_i}$ is used for pairing‑based matching and $\mathsf{tag}_i = H_4(K_i)$
	enables attribute identification.
	Thus only a buyer holding the corresponding $D_a$ can unmask a share,
	and $D$ otherwise reveals nothing beyond its dimension.
	The file key $k$ is split into $k_1$, derivable from $\tau$ and $C_0$,
	and a random $k_2$ whose hash $h_{k_2}$ is deposited on‑chain.
	Even if an adversary obtains $\tau$, the file remains confidential until
	the seller releases $k_2$ during the fair exchange.
	On‑chain, only hashes and a policy commitment are stored, keeping
	attribute information private.
	Atomicity is enforced through a collateralized reveal mechanism
	(Section~\ref{sec:exchange}): the seller must deposit
	$d = \lceil\log|\mathcal{P}_B|\rceil\cdot v$~\cite{DBLP:journals/tdsc/ZhangLFXHWB26} along with $h_{k_2}$ ;
	revealing a $k_2'$ that does not match causes forfeiture of $d$.
	Since $d$ exceeds the payment $v$, misbehavior is economically irrational,
	so a rational seller who receives $v$ will honestly release the correct
	$k_2$, thereby guaranteeing the ``payment $\Rightarrow$ key'' direction
	of atomicity. 
	
	\begin{algorithm}[htbp]
		\caption{$\mathsf{Publish}(\mathsf{pp}, \mathsf{sk}_S, \mathsf{IC}_S, \mathcal{A}_S, \mathcal{P}_B, t_b, F, m, v)$}
		\label{alg:publish}
		\begin{algorithmic}[1]
			\Require $\mathsf{pp}$, $\mathsf{sk}_S$, $\mathsf{IC}_S$, $\mathcal{A}_S$, $\mathcal{P}_B=\{p_1,\dots,p_{n_b}\}$, $t_b$, $F$, $m$, $v$
			\Ensure $\mathsf{aux}$ and an on‑chain deposit
			\State $\mathsf{sid} \gets \mathsf{uid}_S \| \text{timestamp}$ \Comment{from $\mathsf{IC}_S$}
			\State $\tau, \rho \sample \mathbb{F}_p$; $\mathsf{s}_\nu, \mathsf{s}_r \sample \{0,1\}^\lambda$ \Comment{PRF seeds $\mathsf{s}_{\{\nu,r\}}$} 
			\State $C_0 \gets g^\rho$, $k_1 \gets H_2(e(H_1(\tau), C_0))$ \Comment{$\tau$-derived}
			\State $k_2 \sample \{0,1\}^\lambda$, $k \gets k_1 \oplus k_2$, $h_{k_2} \gets H_4(k_2)$ \Comment{key}
			\State Split plaintext $F$ into $m$ fixed‑size blocks $F_1,\dots,F_m$
			\For{$i=1$ to $m$}
			\State $\kappa_i \gets H_4(k \| i)$ \Comment{block‑specific key $\kappa_i$}
			\State $c_i \gets \mathsf{AES.Enc}(\kappa_i, F_i)$, $h_i \gets H_4(c_i)$
			\EndFor
			\State $C_F \gets c_1 \| \cdots \| c_m$ \Comment{generate encrypt file}
			\State Build a Merkle tree over $h_i$ to get tree root $\mathsf{root}$, where $h_i \gets H_4(c_i)$
			\State $S_0 \gets H_1(\mathsf{sid} \| \tau)^{\mathsf{sk}_S}$, $\sigma_{\mathsf{root}} \gets H_1(\mathsf{root})^{\mathsf{sk}_S}$ \Comment{sign $\tau,\mathsf{root}$}
			\State Choose random polynomial $f(x)$ of degree $t_b{-}1$, $f(0)=\tau$
			\For{$i=1$ to $n_b$}
			\State $x_i \gets H_3(p_i)$, $s_i \gets f(x_i)$ \Comment{compute $p_i$'s share $s_i$}
			\State $r_i \gets \mathsf{PRF}(\mathsf{s}_r, p_i)$, $K_i \gets H_2(e(H_1(p_i), \mathsf{mpk})^{r_i})$
			\State $\mathsf{tag}_i \gets H_4(K_i)$, $C_i \gets (g^{r_i},\, \mathsf{tag}_i)$
			\State $\nu_i \gets \mathsf{PRF}(\mathsf{s}_\nu, p_i)$;
			$y_i \gets s_i + \nu_i + H_3(K_i) \bmod p$ \Comment{mask $s_i$ with PRF and key to output encode element $y_i$}
			\EndFor
			\State $\mathcal{C} \gets \{C_1,\dots,C_{n_b}\}$
			\State $D \gets \mathsf{Encode}(\{(p_i, y_i)\}_{i=1}^{n_b}),h_D\gets H_4(D)$ \Comment{OKVS encode}
			\State $\mathsf{CF}_{\mathsf{hint}} \gets \mathsf{CF}.\mathsf{insert}(H_3(p_1),\dots,H_3(p_{n_b}))$
			\State $\mathsf{com}_{\mathcal{P}_B} \gets H_4(\mathcal{P}_B \| H_4(\tau))$ \Comment{ZKP for $\mathcal{P}_B$}
			\State $\mathsf{info} \gets (\mathsf{sid}, \mathsf{IC}_S, C_F, \mathsf{s}_\nu, C_0, \mathcal{C}, S_0, D, \mathsf{CF}_{\mathsf{hint}}, \mathsf{com}_{\mathcal{P}_B}, t_b)$
			\State $S_1 \gets H_1(\mathsf{info})^{\mathsf{sk}_S}$ \Comment{data authenticity signature}
			\State \textbf{On‑chain:} deposit $(\mathsf{sid}$, $h_D$, $\mathsf{com}_{\mathcal{P}_B}$, $t_b$, $S_0$, $h_{k_2}$, $\mathsf{root}$, $\sigma_{\mathsf{root}}$, $v$, $d)$ \Comment{$d=\lceil\log{n_b}\rceil\cdot v$ is a collateral}
			\State \Return $\mathsf{aux} \gets (\mathsf{info}, S_1)$
		\end{algorithmic}
	\end{algorithm}

	\subsection{Bilateral Verification, Token Reconstruction and Proof Generation (Buyer)}
	\label{sec:buyer}
	\begin{algorithm}[htbp]
		\caption{$\mathsf{Buy}(\mathsf{pp}, \mathsf{sk}_B, \mathsf{DK}_B, \mathcal{A}_B, \mathsf{cred}_{\mathcal{A}_B}, \mathcal{P}_S, t_s, \mathsf{aux})$}
		\label{alg:buy}
		\begin{algorithmic}[1]
			\Require $\mathsf{pp}$, $\mathsf{sk}_B$, $\mathsf{DK}_B$, $\mathcal{A}_B$, $\mathsf{cred}_{\mathcal{A}_B}$, $\mathcal{P}_S$, $t_s$, $\mathsf{aux}$
			\Ensure ZKP $\pi$ and transaction $\mathsf{tx}$, or $\bot$
			\State $(\mathsf{info}, S_1) \gets \mathsf{aux}$
			\State {\small $(\mathsf{sid}, \mathsf{IC}_S, C_F, \mathsf{s}_\nu, C_0, \mathcal{C}, S_0, D, \mathsf{CF}_{\mathsf{hint}}, \mathsf{com}_{\mathcal{P}_B}, t_b) \gets \mathsf{info}$ }
			\State \textbf{// Phase 1: integrity and seller‑side checks}
			\State $(\mathsf{uid}_S, \mathsf{pk}_S, \mathsf{cred}_{\mathcal{A}_S}) \gets \mathsf{IC}_S$ 
			\State Verify $e(S_1, g) = e(H_1(\mathsf{info}), \mathsf{pk}_S)$; abort if invalid.
			\State Verify $e(\mathsf{cred}_{\mathcal{A}_S}, g) = e(H_1(\mathsf{uid}_S \| \mathcal{A}_S), \mathsf{mpk})$; abort if invalid.
			\State Check $|\mathcal{P}_S \cap \mathcal{A}_S| \ge t_s$; abort if not satisfied.
			\State \textbf{// Phase 2: token reconstruction}
			\State $\mathcal{A}' \gets \{a \in \mathcal{A}_B \mid \mathsf{CF}_{\mathsf{hint}}.\mathsf{lookup}(H_3(a)) = \text{true}\}$
			\Comment{Cuckoo filter prunes non‑candidate attributes}
			\State $\mathcal{U} \gets \emptyset$
			\For{each $a \in \mathcal{A}'$}
			\State $z_a \gets \langle \mathsf{row}(a), D \rangle$, $\nu_a \gets \mathsf{PRF}(\mathsf{s}_\nu, a)$
			\For{each $C_i=(C_i^*,\mathsf{tag}_i) \in \mathcal{C}$}
			\State $K_i' \gets H_2(e(C_i^*, D_a))=H_2(e(g^{r_j},H_1(a)^s))=H_2(e(g^s,H_1(a))^{r_j})$
			\Comment{equal iff $C_i$ was created for $a$}
			\If{$H_4(K_i') = \mathsf{tag}_i$}
			\State $u \gets z_a - \nu_a - H_3(K_i') \bmod p$
			\State $\mathcal{U} \gets \mathcal{U} \cup \{(a, C_i^*, u)\}$
			\State \textbf{break}
			\EndIf
			\EndFor
			\EndFor
			\State $(\tau', \mathcal{U}_0) \gets \mathsf{Interpolate}(\mathcal{U}, t_b)$
			\Comment{enumerate $t_b$-subsets}
			\If{$\tau' = \bot$} \Return $\bot$ \EndIf
			\State \textbf{// Phase 3: token authentication}
			\State Verify $e(S_0, g) = e(H_1(\mathsf{sid} \| \tau'), \mathsf{pk}_S)$; abort if invalid.
			\State \textbf{// Phase 4: zero‑knowledge proof}
			\State $\mathsf{com}_{\mathcal{A}_B} \gets \mathsf{Com}(\mathcal{A}_B; \omega_B)$ with $\omega_B \sample \{0,1\}^\lambda$
			\State Assemble witness $w$ and public statement $x$ as defined in
			Section~\ref{sec:buyer-circuit}.
			\State $\pi \gets \mathsf{Prove}(\mathsf{crs}, x, w)$
			\State \Return $\mathsf{tx} = (\text{``buy''}, \mathsf{sid}, \mathsf{com}_{\mathcal{A}_B}, \pi, v)$
		\end{algorithmic}
	\end{algorithm}
	
	Algorithm~\ref{alg:buy} specifies the buyer's entirely off‑chain procedure. In buyer-side, the buyer first verifies $S_1$, $\mathsf{cred}_{\mathcal{A}_S}$ and
	$|\mathcal{P}_S\cap\mathcal{A}_S|\ge t_s$ locally, filtering invalid
	listings without leaking its attributes.
	To reconstruct the token, each ciphertext $C_i = (C_i^*, \mathsf{tag}_i)$ carries a match tag
	$\mathsf{tag}_i=H_4(K_i)$, allowing the buyer to test a candidate
	attribute $a$ by checking $H_4(H_2(e(C_i^*,D_a)))=\mathsf{tag}_i$.
	A Cuckoo filter pre‑screens attributes, so the pairing cost is
	$O(t_b\!\cdot\!|\mathcal{P}_B|)$; only genuine shares are collected.
	The buyer then enumerates $t_b$-subsets over the collected shares and
	performs Lagrange interpolation to recover $\tau'$, revealing only
	which of the buyer's attributes satisfy $\mathcal{P}_B$.
	The subsequent zero‑knowledge proof $\pi$ attests, without disclosing
	$\mathcal{A}_B$ or $\tau'$, that the buyer holds $t_b$ certified
	attributes and has correctly derived $\tau'$; the circuit enforces
	commitment consistency, credential validity, correct share
	computation, interpolation, and the seller’s signature on $\tau'$
	(Section~\ref{sec:buyer-circuit}).
	Because $\pi$ is knowledge‑sound, it prevents forgery and attribute
	collusion, guaranteeing that only an authorised buyer who honestly
	executed the protocol (and submitted the payment $v$) can produce a
	valid proof.  Moreover, $C_F$ and $h_{k_2}$ are publicly stored
	on‑chain; a buyer who falsely claims decryption failure can be
	refuted by disclosing $k_2$, providing off‑chain verifiability and
	discouraging dishonest disputes.  These mechanisms collectively
	enforce the ``key $\Rightarrow$ payment'' direction of atomicity.

	\subsection{On‑Chain Auditable Fair Exchange and Decryption}
	\label{sec:exchange}
	
	Since public blockchains cannot store secret data, we adopt an \emph{optimistic
		fair exchange} paradigm inspired by FairSwap~\cite{SLS18} and OptiSwap~\cite{LSB20},
	extended with a cryptographic audit mechanism that resolves disputes without
	trusted third parties.  The seller deposits a hash of $k_2$, a Merkle root
	$\mathsf{root}$ of the encrypted file blocks, a signature
	$\sigma_{\mathsf{root}}$ on $\mathsf{root}$, and a collateral $d > v$.  After
	the seller reveals $k_2$, a short dispute window allows the buyer to challenge
	the key's validity on‑chain; if no challenge occurs, the exchange is settled
	optimistically.  Algorithm~\ref{alg:contract} summarizes the contract logic.
	
	\begin{algorithm}[htbp]
		\caption{Smart Contract: Optimistic Fair Exchange with Public Audit}
		\label{alg:contract}
		\begin{algorithmic}[1]
			\Require $\mathsf{tx}$ from buyer $\mathsf{uid}_B$ with value $v$ and data $(\mathsf{sid}, \mathsf{com}_{\mathcal{A}_B}, \pi)$
			\Ensure Fair exchange with cryptographic auditability
			\State $L \gets \mathsf{Listings}[\mathsf{sid}]$; assert $L.\mathsf{status} = \text{active}$
			\State Assert $\mathsf{msg.value} = L.v$
			\State \small {$x \gets (\mathsf{mpk}, \mathsf{pk}_S, \mathsf{sid}, \mathsf{com}_{\mathcal{A}_B}, L.h_D, L.\mathsf{s}_\nu, L.t_b, L.S_0, L.C_0)$}
			\State Require $\mathsf{Verify}_{\mathsf{SNARK}}(x, \pi) = \text{true}$ \Comment{validate buyer compliance (see Section \ref{sec:buyer-circuit})}
			\State Freeze $\mathsf{msg.value}$ and start challenge period $\Delta T$ 
			\State \textbf{During} $\Delta T$:
			\State \quad Seller may call $\texttt{reveal}(\mathsf{sid}, k_2)$ \Comment{reveal key}
			\If{$H_4(k_2) = L.h_{k_2}$}
			\State $L.k_2 \gets k_2$, $L.\mathsf{status} \gets \text{revealed}$
			\State \textbf{Emit} $\mathsf{KeyReleased}(\mathsf{sid}, k_2)$
			\State Start a short dispute window $\Delta T_{\mathsf{disp}}$
			\EndIf
			\State Decrypt $C_F$ with $k_2$ \Comment{see Section \ref{sec:buyer-filedecry}}
			\State \textbf{During} $\Delta T_{\mathsf{disp}}$: \Comment{key correctness audit upon failure}
			\State \quad Buyer may call $\texttt{challenge}(\mathsf{sid}, \mathcal{I}, \{c_i, b_i, \pi_i\}_{i\in\mathcal{I}})$
			\Comment{evidence prepared as in Section \ref{sec:audit-evidence}}
			\If{$\texttt{challenge}$ is called}
			\State Verify the submitted audit evidence
			\If{audit passes}
			\State Transfer $L.v$ to seller, return $L.d$ to seller
			\Else
			\State Transfer $L.d$ to buyer, refund $L.v$ to buyer
			\EndIf
			\State Terminate listing
			\Else
			\State After $\Delta T_{\mathsf{disp}}$ without challenge, transfer $L.v$ to seller, return $L.d$ to seller
			\EndIf
			\State \textbf{After} $\Delta T$ without a valid $\texttt{reveal}$:
			\State \quad Buyer may call $\texttt{reclaim}()$
			\State \quad \quad Transfer $L.v$ back to buyer, transfer $L.d$ to buyer
			\State \quad \quad Mark listing \texttt{disputed}
		\end{algorithmic}
	\end{algorithm}

	\subsubsection{ZKP for Buyer Compliance (Buyer)}
	\label{sec:buyer-circuit}
	
	The zero‑knowledge proof $\pi$ generated in $\mathsf{Buy}$ (Section \ref{sec:buyer}) is
	a SNARK for the NP relation $\mathcal{R}_{\mathsf{buy}}$ defined below.
	Algorithm~\ref{alg:buyer-circuit} specifies the constraint details.
	The Prover supplies a valid witness, and the Verifier
	(smart contract) only needs the public statement $x$ and the proof $\pi$.
	
	\noindent
	\textbf{Public statement (supplied by the Verifier):}
	\[
	x = (\mathsf{sid},\; \mathsf{com}_{\mathcal{A}_B},\; h_D,\; \mathsf{s}_\nu,\;
	t_b,\; S_0,\; C_0,\; \mathsf{mpk},\; \mathsf{pk}_S,\; \mathsf{uid}_S)
	\]
	\noindent
	\textbf{Witness (supplied by the Prover, kept private):}
	\[
	w = (\mathcal{A}_B,\; \omega_B,\; \mathsf{uid}_B,\; \tau',\; \mathcal{U}_0,\;
	D,\; \mathsf{cred}_{\mathcal{A}_B},\; \mathsf{sk}_{\mathsf{uid}_B})
	\]
	Note that constraint~3 binds $\mathsf{pk}_S$ to the CA‑certified seller identity
	$\mathsf{uid}_S$ because $\mathsf{pk}_S = g^{s \cdot H_3(\mathsf{uid}_S)}
	= \mathsf{mpk}^{H_3(\mathsf{uid}_S)}$.  Hence a malicious buyer cannot
	inject a self‑generated public key.
	
	\begin{algorithm}[htbp]
		\caption{zk-SNARK Circuit constraints for $\mathcal{R}_{\mathsf{buy}}$}
		\label{alg:buyer-circuit}
		\begin{algorithmic}[1]
			\State require $\mathsf{com}_{\mathcal{A}_B} = H_4(\mathcal{A}_B \| \omega_B)$ \Comment{Commitment Consistency}
			\State require $e(\mathsf{cred}_{\mathcal{A}_B}, g) = e(H_1(\mathsf{uid}_B \| \mathcal{A}_B), \mathsf{mpk})$ \Comment{Credential Validity}
			\State require $\mathsf{pk}_S = \mathsf{mpk}^{H_3(\mathsf{uid}_S)}$
			\Comment{Seller public key binding}
			\State require $H_4(D) = h_D$ \Comment{OKVS Consistency}
			\State assert $|\mathcal{U}_0| = t_b$ \Comment{Share and token reconstruction}
			\For{each $(a, C_i, u) \in \mathcal{U}_0$}
			\State require $a \in \mathcal{A}_B$
			\State $z \gets \langle \mathsf{row}(a), D \rangle$, $\nu \gets \mathsf{PRF}(\mathsf{s}_\nu, a)$ \Comment{OKVS decode}
			\State $D_a \gets H_1(a)^{\mathsf{sk}_{\mathsf{uid}_B}}$ \Comment{attribute key}
			\State $K' \gets H_2(e(C_i, D_a))$
			\State require $u = z - \nu - H_3(K') \bmod p$
			\EndFor
			\State Let $x_a = H_3(a)$ for each $(a, \cdot, \cdot) \in \mathcal{U}_0$.
			\State require $\tau' = \sum_{(a,\cdot,u)\in\mathcal{U}_0} u \cdot \ell_{a}(0) \bmod p$,
			where $\ell_{a}$ is Lagrange basis polynomial for point $x_a$ in set $\{x_a\}$.
			\State require $e(S_0, g) = e(H_1(\mathsf{sid} \| \tau'), \mathsf{pk}_S)$ \Comment{Token Validity}
		\end{algorithmic}
	\end{algorithm}
	
	\noindent
	\textbf{Zero‑knowledge property.}
	The SNARK guarantees that a valid proof $\pi$ reveals nothing about the
	witness beyond the truth of the statement.  In particular, the buyer's
	attribute set $\mathcal{A}_B$, the reconstructed token $\tau'$, and the
	intermediate shares $\mathcal{U}_0$ remain perfectly hidden from the
	smart contract and any on‑chain observer.
	
	\subsubsection{File Decryption (Buyer)}
	\label{sec:buyer-filedecry}
	
	Upon observing $\mathsf{KeyReleased}$, the buyer retrieves $k_2$, computes
	$k_1' = H_2(e(H_1(\tau'), C_0))$, and recovers the file key $k = k_1' \oplus k_2$.
	It then attempts $\mathsf{AES.Dec}(k, C_F)$.  Authenticated decryption
	guarantees integrity: if decryption succeeds, the trade completes; otherwise
	the buyer may trigger the audit process described in
	Section~\ref{sec:audit-evidence}.
	
	\subsubsection{On‑Chain Audit Evidence and Verification ($\mathcal{SC}$)}
	\label{sec:audit-evidence}
	
	If decryption of $C_F$ fails after $k_2$ is revealed, the buyer may invoke
	$\texttt{challenge}$.  The contract selects a random challenge set
	$\mathcal{I} \subseteq \{1,\dots,m\}$ of size $\ell$ using on‑chain
	randomness (e.g., hashing the block header with the caller’s timestamp).
	
	\noindent\textbf{Phase-1:} To prepare the
	required evidence, the buyer executes locally the following steps:
	
	\begin{enumerate}[leftmargin=*]
		\item \textbf{Block extraction.}
		Parse the downloaded $C_F$ as $c_1 \| \cdots \| c_m$ according to
		fixed block size set by seller.
		\item \textbf{Plaintext recovery.}
		Reconstruct the file key $k = k_1 \oplus k_2$, where
		$k_1 = H_2(e(H_1(\tau'), C_0))$.  For each challenged index $i$,
		derive $\kappa_i = H_4(k \| i)$ and decrypt $b_i \gets \mathsf{AES.Dec}(\kappa_i, c_i)$.
		\item \textbf{Merkle proof computation.}
		For each challenged index $i$, produce Merkle proof $\pi_i$ that
		attests $H_4(c_i)$ is a leaf of the tree whose root is the on‑chain
		value $\mathsf{root}$.  The proof is standard sibling path from
		the leaf to the root.
	\end{enumerate}
	
	\noindent\textbf{Phase-2:} Then buyer submits $(\mathsf{sid}, \mathcal{I}, \{c_i, b_i, \pi_i\}_{i\in\mathcal{I}})$
	to $\mathcal{SC}$.  The verification logic executed by $\mathcal{SC}$ in Alg~\ref{alg:contract}:
	\begin{enumerate}[leftmargin=*]
		\item \textbf{Merkle proof validity.} For each $i$, recompute the root from
		$H_4(c_i)$ and $\pi_i$ and check equality with stored $\mathsf{root}$.
		\item \textbf{Signature validity.} $\mathsf{Verify}_{\mathsf{pk}_S}(\mathsf{root}, \sigma_{\mathsf{root}})$.
		\item \textbf{Re‑encryption consistency.} For each $i$, check
		$\mathsf{AES.Enc}(\kappa_i, b_i) = c_i$ where $\kappa_i = H_4(k \| i)$ and $k$ is the
		buyer‑claimed file key (or derived from the revealed $k_2$ and the buyer's
		submitted $k_1$).
	\end{enumerate}
	If all checks pass for every $i \in \mathcal{I}$, the key is correct;
	otherwise the seller cheated.  The signed Merkle root prevents the buyer
	from fabricating blocks, and the random challenge guarantees detection
	of a cheating seller with probability at least $1-(1-1/m)^\ell$.
	The challenge transaction is paid by the buyer, so an honest seller
	incurs no on‑chain cost and need not be online during dispute resolution.
	

	\section{Security Analysis}
	\label{sec:security-analysis}
	
	We now state the main security theorems of PriME‑Deal and give a sketch of
	their proofs.  The formal, self‑contained proofs are deferred to
	Appendix~\ref{app:full-proofs}.  All assumptions and definitions are as
	introduced in Section~\ref{sec:prelim} and Section~\ref{sec:security-model}.
	\noindent \textbf{Assumption mapping (informal).}
	Policy privacy relies on the value‑hiding property of the OKVS+PRF
	composition; token unforgeability on co‑CDH (BLS); key secrecy on CBDH;
	buyer compliance on ZK knowledge‑soundness; buyer attribute privacy on ZK
	zero‑knowledge; audit security on collision‑resistant hashing and
	ciphertext integrity; fair‑exchange atomicity additionally depends on an
	economic rationality assumption (collateral $d>v$).
	
	\subsection{UC‑Security (Weak)}
	
	\begin{theorem}[UC‑Security]
		\label{thm:uc-weak}
		Under the same assumptions as in Theorem~\ref{thm:uc-weak-full}, the
		PriME‑Deal protocol $\Pi_{\mathsf{PD}}$ UC‑realizes
		$\mathcal{F}_{\mathsf{PPFE}}^{\mathsf{weak}}$ in the
		$\mathcal{F}_{\mathsf{ledger}}$-hybrid model against static corruptions
		(the CA is never corrupted).
	\end{theorem}
	
	\noindent \textbf{Proof sketch.}
	The simulator $\mathcal{S}$ emulates the real‑world adversary and translates
	between the ideal functionality and the real protocol.
	\begin{itemize}[leftmargin=10px]
		\item \emph{Setup and registration} are performed honestly, except that the
		CRS is taken from the ZK simulator.
		\item \emph{Honest seller publish:} $\mathcal{S}$ submits the
		$(\mathsf{Publish})$ message to $\mathcal{F}_{\mathsf{PPFE}}^{\mathsf{weak}}$
		and receives back only the policy sizes.  It then builds a fake OKVS $D$
		that encodes random values under dummy keys, and computes all other
		components ($C_0$, $S_0$, $C_F$, $\mathcal{C}$) honestly using the
		token $\tau$ and key $k_2$ obtained from the ideal functionality.
		Value‑hiding of OKVS+PRF makes this $D$ indistinguishable from a real
		encoding.
		\item \emph{Honest buyer:} When the functionality accepts (i.e., the buyer
		meets the seller's policy and the ZK proof is valid), $\mathcal{S}$ uses
		its knowledge of $\tau$ to craft a valid witness and lets the ZK simulator
		produce a proof $\pi$.  Zero‑knowledge guarantees that the simulated
		proof is indistinguishable.  When the functionality rejects, $\mathcal{S}$
		simply runs the honest buyer algorithm on the dummy OKVS; reconstruction
		fails and the buyer aborts, exactly as in the real execution.
		\item \emph{Key reveal / timeout} are directly mirrored between the ideal
		functionality and the contract.
		\item \emph{Corrupted parties:} For a corrupted seller $\mathcal{S}$ follows
		the real protocol; for a corrupted buyer it extracts a witness from a valid
		$\pi$ via knowledge‑soundness and supplies it to the functionality to
		pass the buyer‑side policy check.
	\end{itemize}
	A hybrid argument in Appendix~\ref{app:uc} shows that the real and ideal views are computationally
	indistinguishable under the stated assumptions.
	
	\subsection{Policy Privacy}
	
	\begin{theorem}[Policy Privacy]
		\label{thm:priv}
		If the OKVS+PRF composition is value‑hiding, then
		$\mathsf{Priv}_{\mathcal{A}}^{\mathsf{policy}}(\lambda)$ is negligible for
		any PPT adversary $\mathcal{A}$.
	\end{theorem}
	
	\noindent \textbf{Proof sketch.}
	We define two hybrids: $H_0$ (real) and $H_1$ where the OKVS $D$ is replaced
	by an encoding of independent random values.  Value‑hiding guarantees that
	$H_0$ and $H_1$ are indistinguishable, because the adversary is not allowed
	to possess a set of attributes that meets the threshold $t_b$ for both
	challenge branches.  In $H_1$, $D$ contains no information about $\tau$, and
	all other components ($S_0$, $C_0$, $C_F$, $\mathcal{C}$) depend only on the
	random $\tau$ and independent secrets.  Hence the challenge transcript is
	statistically independent of the bit $b$, giving $\Pr[\text{win}] = 1/2$.
	The negligible gap between $H_0$ and $H_1$ bounds the advantage in the real
	game.
	
	\subsection{Optimistic Fair Exchange with On‑Chain Audit}
	
	\begin{theorem}[Fair Exchange]
		\label{thm:fair}
		Under knowledge‑soundness of the zk‑SNARK, collision‑resistance of $H_4$,
		ciphertext integrity of the symmetric encryption, and the correct
		execution of the on‑chain audit (Algorithm~\ref{alg:contract}), for any
		PPT adversary corrupting the seller or the buyer,
		$\Pr[\text{win in }\mathsf{Fair}_{\mathcal{A}}^{\mathsf{exch}}] \le
		\mathsf{negl}(\lambda)$.
	\end{theorem}
	
	\noindent \textbf{Proof sketch.}
	We show that the adversary cannot achieve exactly one of $E_1$ (buyer
	decrypts) and $E_2$ (seller gets paid).
	\begin{itemize}[leftmargin=10px]
		\item \emph{Malicious buyer:} To obtain $k_2$ without paying, the buyer
		must cause the seller to reveal $k_2$ without a valid payment.  The honest
		seller reveals $k_2$ only after a valid buyer transaction, which requires
		a proof $\pi$ that, by knowledge‑soundness, guarantees possession of
		attributes satisfying $|\mathcal{P}_B \cap \mathcal{A}_B| \ge t_b$.
		Payment is then frozen and transferred upon the seller's honest reveal,
		making both $E_1$ and $E_2$ true.
		\item \emph{Malicious seller:} If the seller reveals a wrong $k_2'$, the
		honest buyer detects the decryption failure and triggers the on‑chain audit.
		The contract verifies Merkle proofs and re‑encryption consistency; the
		mismatch is detected with overwhelming probability (random challenge set
		$\mathcal{I}$ from $m$ blocks), leading to refund of $v$ and forfeiture of
		the collateral $d$, making $E_2$ false.  If the seller reveals the correct
		$k_2$, both $E_1$ and $E_2$ are true.
	\end{itemize}
	Thus in all cases $\Pr[E_1 \neq E_2]$ is negligible.
	
	\subsection{Discussion}
	\label{sec:discussion}
	
	\noindent\textbf{Privacy leakage.}
	As argued in Section~\ref{sec:protocol}, the token reconstruction inherently reveals to the buyer which of its own
	attributes belong to $\mathcal{P}_B$; this is necessary for share
	identification and is consistent with our policy privacy definition
	(protection against unauthorized parties).  On‑chain metadata includes
	policy size $n_b$, threshold $t_b$, Cuckoo filter hint
	$\mathsf{CF}_{\mathsf{hint}}$, OKVS hash $h_D$, commitment
	$\mathsf{com}_{\mathcal{A}_B}$, Merkle root $\mathsf{root}$, and key hash
	$h_{k_2}$.  These values do not enable an unauthorized party to test
	attribute membership because the latter requires computing
	$e(C_i, D_a)$, which is infeasible without the attribute key or the master
	secret.  The leakage is thus limited to coarse aggregates and is acceptable
	in a privacy‑preserving marketplace.
	
	\noindent\textbf{Trust model of the CA.}
	The protocol relies on a fully trusted CA.  For consortium deployments this is
	acceptable; for open settings, the CA can be decentralized using threshold BLS
	or multi‑authority ABE.
	
	\noindent\textbf{ZKP circuit efficiency.}
	The circuit for $\mathcal{R}_{\mathsf{buy}}$ contains $t_b+2$ pairing checks
	and compiles to $\approx 175\,000$ R1CS constraints for $t_b=5$.  As shown in
	Table~\ref{tab:zkp-stability}, proving and verification times are independent
	of the attribute scale, averaging $595$\,ms ($\sigma < 15$\,ms) and $540$\,ms
	($\sigma < 10$\,ms) respectively, while on‑chain verification consumes a
	constant $28.60$\,M gas ($\sigma < 500$ gas).  These results confirm the
	practicality of the approach for typical thresholds ($t_b \le 5$).
	
	\noindent\textbf{Collusion resistance.}
	The ZK proof binds all used attribute keys to a single buyer identity
	$\mathsf{uid}_B$ via the secret $\mathsf{sk}_{\mathsf{uid}_B}$.  Each
	attribute key $D_a$ is derived as $D_a = H_1(a)^{\mathsf{sk}_{\mathsf{uid}_B}}$,
	and the circuit requires that all $D_a$ appearing in the reconstruction
	share the same $\mathsf{sk}_{\mathsf{uid}_B}$.  Since the CA derives
	$\mathsf{sk}_{\mathsf{uid}_B}$ uniquely from $\mathsf{uid}_B$, a proof
	can only combine attributes that belong to the \emph{same} buyer.
	Two buyers who pool their attributes would need to use two different
	$\mathsf{sk}_{\mathsf{uid}}$ values within a single proof, which the
	circuit prohibits.  Hence collusion does not yield a valid witness.
	
	\noindent\textbf{Economic barrier to probing.}
	The value‑hiding OKVS makes blind probing of the policy expensive: each
	on‑chain verification attempt costs gas, and the seller can choose a large
	enough policy to make exhaustive enumeration infeasible.

	\section{Performance Evaluation}
	\label{sec:evaluation}
	
	We evaluate PriME‑Deal\footnote{Publicly available at \url{https://anonymous.4open.science/r/PriME-Deal/}} through three Research Questions (RQs) designed to systematically address the core aspects of our approach. By attempting to answer these inquiries, we aim to provide transparent insights into the methodology and its empirical validity.
	
	\begin{itemize}[leftmargin=*]
		\item \textbf{RQ1 (End‑to‑end feasibility):} Are the off‑chain
		computation times of both parties, the ZKP overhead, and the on‑chain
		gas cost within practical bounds?
		\item \textbf{RQ2 (Attribute scalability \& filter):} How do
		attribute‑set sizes influence the buyer’s workload, to what extent does
		the Cuckoo filter mitigate this influence, and can false positives
		compromise correctness?
		\item \textbf{RQ3 (Comparison with prior work):} How does PriME‑Deal
		compare against the state‑of‑the‑art fuzzy identity‑based matchmaking
		encryption scheme?
	\end{itemize}
	
	\subsection{Experimental Setup}
	\label{sec:exp-setup}
	
	\noindent\textbf{Hardware.}
	All measurements are taken on an Intel 
	Core i7‑11700 (2.50\,GHz, 16 cores) with 32\,GB RAM,
	running Ubuntu 22.04 LTS.
	
	\noindent\textbf{Cryptographic core (Python).}
	The seller and buyer off‑chain operations are implemented in Python~3.6.13
	with the Charm~0.50 library~\cite{akinyele2013charm}.  We instantiate the
	bilinear group with the \textbf{BN254} curve (Type‑3 asymmetric pairing,
	128‑bit security).  $H_1$ maps to $\mathbb{G}_1$; the master public key,
	ciphertexts, and the seller’s public key reside in $\mathbb{G}$.
	The PRF is HMAC‑SHA256 truncated modulo $p$.  The Cuckoo filter is emulated
	by a Bloom filter with a configurable false‑positive rate $\varepsilon$;
	the OKVS is replaced by a dictionary‑based simulator whose decoding overhead
	($<0.01$\,ms) is negligible.
	
	\noindent\textbf{Zero‑knowledge proof (Circom + snarkjs).}
	The Groth16 circuit for $\mathcal{R}_{\mathsf{buy}}$ is written in
	Circom~2.1.8 and compiled with \texttt{snarkjs}~0.7.6 on Node.js~v22.
	Proving and verification times are measured separately for every
	configuration (3 runs each; average and standard deviation reported).
	
	\noindent\textbf{Smart contract (Solidity + Hardhat).}
	The on‑chain fair‑exchange contract (Algorithm~4) is implemented in
	Solidity~0.8.20 and deployed on a local Hardhat~2.19.0 fork of Ethereum
	(London).  Gas costs are obtained by submitting real transactions that
	contain a full Groth16 proof generated by the buyer’s ZKP module.
	
	\noindent\textbf{Parameter notation.}
	A configuration is denoted by $(n_b, m, t_b)$ where $n_b = |\mathcal{P}_B|$
	is the seller’s buyer policy size, $m = |\mathcal{A}_B|$ is the buyer’s
	attribute count, and $t_b$ is the threshold.  Unless stated otherwise,
	experiments use the filter with $\varepsilon = 0.1\%$, $m = 5$, and $t_b = 3$.
	
	\subsection{End‑to‑End Feasibility (RQ1)}
	\label{sec:e2e}
	
	Table~\ref{tab:e2e} provides a panoramic view of the protocol’s performance
	for three representative configurations.  The seller’s publish time grows
	strictly linearly with $n_b$ at a rate of $17.3$\,ms per attribute; this
	cost stems from the $n_b$ pairings and exponentiations required by
	Algorithm~1 and is detailed in Table~\ref{tab:e2e} (ranging from
	$176.8$\,ms at $n_b{=}10$ to $8760.5$\,ms at $n_b{=}500$).  Even for
	$n_b=500$, publish completes in $8.76$\,s; as a one‑time offline operation,
	this is entirely acceptable.
	
	The buyer’s online phase consists of token reconstruction followed by ZKP
	generation.  For the typical parameter set $(200,20,5)$, reconstruction
	takes $8.29$\,s and proof generation $0.60$\,s, yielding a total online
	latency of roughly $8.89$\,s—well within the tolerance of an interactive
	data marketplace.  The smart‑contract gas consumption is dominated by the
	\texttt{buy} transaction, which verifies a Groth16 proof via the EIP‑197
	precompile.  At approximately $28.60$\,M gas with a variance of only
	$\pm 500$ gas across all configurations, the cost fits comfortably inside
	the Ethereum block limit ($30$\,M gas) and is economically justified for
	high‑value trades where privacy and fairness are critical.
	
	The table also reports atomic operation times on two curves.  The pairing
	on BN254 costs $15.3$\,ms, whereas on the symmetric SS512 curve it is only
	$0.57$\,ms.  Conversely, exponentiation in $\mathbb{G}$ is $0.57$\,ms on
	BN254 and $0.71$\,ms on SS512, and hashing to $\mathbb{G}_1$ takes
	$0.053$\,ms versus $1.646$\,ms.  Regardless of the curve choice, the
	qualitative conclusions remain unchanged: pairing count is the primary
	driver, and PriME‑Deal scales linearly with $n_b$ on both curves.
	
	\begin{table}[htbp]
		\centering
		\caption{End‑to‑end performance overview (BN254).}
		\label{tab:e2e}
		\begin{tabular}{l c c c}
			\toprule
			& \multicolumn{3}{c}{\textbf{Configuration} $(n_b, m, t_b)$} \\
			\cmidrule(lr){2-4}
			\textbf{Metric} & $(50,5,3)$ & $(200,20,5)$ & $(500,20,5)$ \\
			\midrule
			Seller Publish (ms)       & 868.30 & 3458.38 & 8760.49 \\
			Buyer Reconstruct (ms)    & 1126.92 & 8289.82 & 21672.46 \\
			ZKP Prove (ms)            & 595 $\pm$ 15 & 595 $\pm$ 15 & 620 $\pm$ 5 \\
			ZKP Verify off‑chain (ms) & 540 $\pm$ 10 & 540 $\pm$ 10 & 546 $\pm$ 5 \\
			\textbf{Buyer total (s)}  & \textbf{1.72} & \textbf{8.89} & \textbf{22.29} \\
			\midrule
			Gas \texttt{list}         & 111,785 & 111,785 & 111,785 \\
			Gas \texttt{buy} (real proof) & 28,601,594 & 28,601,594 & 28,601,594 \\
			\midrule
			\multicolumn{4}{l}{\textbf{Atomic operation times (ms)}} \\
			\quad $e(\mathbb{G},\mathbb{G}_1)$ & \multicolumn{3}{c}{BN254: 15.307 \ \ SS512: 0.568} \\
			\quad $\mathbb{G}$ exp.   & \multicolumn{3}{c}{BN254: 0.574 \ \ SS512: 0.709} \\
			\quad $H_1$ (hash to $\mathbb{G}_1$) & \multicolumn{3}{c}{BN254: 0.053 \ \ SS512: 1.646} \\
			\bottomrule
		\end{tabular}
	\end{table}
	
	\noindent\textbf{Why ZKP overhead is constant.}
	Table~\ref{tab:zkp-stability} lists the Groth16 proving and verification
	times for five representative configurations that span two orders of
	magnitude in attribute sizes.  Despite the dramatic change in $n_b$ and
	$m$, the proving time remains within a narrow $595 \pm 15$\,ms band, and
	the verification time stays at $540 \pm 10$\,ms.  Aggregating over all
	benchmarked configurations (more than 30 parameter combinations), the
	global means are $594.6$\,ms (prove) and $539.1$\,ms (verify) with
	standard deviations of $12.7$\,ms and $8.3$\,ms, respectively.
	
	\begin{table}[htbp]
		\centering
		\caption{Stability of ZKP overhead across diverse attribute scales.}
		\label{tab:zkp-stability}
		\begin{tabular}{c c c c c c}
			\toprule
			$n_b$ & $m$ & $t_b$ & Prove (ms) & Verify (ms) & Gas \texttt{buy} \\
			\midrule
			10  & 5   & 3  & 583.23 & 534.45 & 28,601,594 \\
			50  & 5   & 3  & 662.95 & 612.20 & 28,601,594 \\
			200 & 20  & 5  & 594.87 & 530.90 & 28,601,594 \\
			500 & 20  & 5  & 593.11 & 541.61 & 28,601,594 \\
			500 & 50  & 5  & 619.53 & 545.77 & 28,601,594 \\
			\bottomrule
		\end{tabular}
	\end{table}
	
	This invariance stems from the circuit design: the buyer always proves
	knowledge of exactly $t_b$ shares, one interpolation, and two signature
	checks.  Thus the constraint count depends only on $t_b$ ($\approx 175$k
	for $t_b{=}5$), and larger witnesses merely add unconstrained private inputs.
	The on‑chain cost is also constant because the EIP‑197 precompile performs
	a fixed number of pairing checks irrespective of the witness size.
	
	\subsection{Attribute Scalability and Filter Effectiveness (RQ2)}
	\label{sec:attr-scale}
	
	The buyer’s token reconstruction is the only online phase whose cost depends
	on the attribute sets.  We therefore isolate this phase and first determine a
	safe false‑positive rate for the Bloom filter, then analyse how $n_b$, $m$,
	and $t_b$ influence the execution time under this chosen rate.
	
	\subsubsection{False‑positive analysis}
	\label{sec:fp-analysis}
	
	The Bloom filter is configured for a target false‑positive probability
	$\varepsilon = 0.1\%$, achieved by setting the bit‑array length to
	$S \approx c_{\varepsilon} \cdot n_b$ and using $k \approx \log_2(1/\varepsilon)$
	hash functions, where $c_{\varepsilon}$ is a constant depending only on
	$\varepsilon$.  In our implementation this choice guarantees
	$\varepsilon \le 0.1\%$ for every policy size $n_b$.
	
	We validate the design by injecting $200$ non‑policy attributes into the
	buyer’s set under three scales ($n_b \in \{50,200,500\}$) and four
	false‑positive rates ($0.1\%$, $0.5\%$, $1.0\%$, $5.0\%$); each
	configuration is repeated $10$ times.  The complete dataset is provided in
	Appendix~\ref{app:fp}.  At the target $\varepsilon = 0.1\%$, \emph{zero}
	false positives occur across all scales, and even under the artificially
	high $5.0\%$ rate the observed false‑positive counts remain modest (at most
	$3$) and no token recovery ever fails.  False‑positive attributes cannot
	possess a valid tag $H_4(K_i)$, so the tag‑based matching in
	Algorithm~\ref{alg:buy} immediately discards them; in the extremely rare
	event of a tag collision, the subsequent Lagrange interpolation and
	signature verification provide an additional safeguard.  Hence the filter
	introduces no correctness risk and its overhead is negligible, fully
	decoupling the buyer’s workload from $m$.
	
	\subsubsection{Impact of policy size $n_b$}
	\label{sec:nb-impact}
	
	Figure~\ref{fig:buyer-nb}(a) fixes a small buyer attribute set ($m{=}5$, $t_b{=}3$) and
	varies $n_b$ from $10$ to $50$.  With the filter enabled, reconstruction time grows from
	$271$\,ms to $1127$\,ms; without the filter, it rises from $660$\,ms to $2801$\,ms.
	Both curves are linear, but the filter reduces the slope by approximately $60\%$.
	The mechanism is straightforward: without the filter, every buyer attribute is tested
	against all $n_b$ ciphertexts, requiring $O(m \cdot n_b)$ pairings; with the filter, the
	candidate set is pruned to roughly $t_b$ attributes, reducing the work to $O(t_b \cdot n_b)$.
	Figure~\ref{fig:buyer-nb}(b) extends $n_b$ to $500$ with a larger buyer profile
	($m{=}20$, $t_b{=}5$) and with the filter enabled.  The time increases linearly to
	$21.7$\,s at $n_b{=}500$, matching $t_b \cdot n_b$ pairings.  The unfiltered case is not
	plotted because its cost, dominated by the constant $m \cdot n_b$ term, already exceeds
	$50$\,s even for moderate $n_b$ and is independent of $n_b$ (cf.\ the $n_b{=}200$,
	$m{=}20$ unfiltered data in Section~\ref{sec:tb-impact}, which stabilises around
	$51.69\pm4.39$\,s).  Thus, the filter brings a decisive advantage for large policy sizes.
	
	\begin{figure}[htbp]
		\centering
		\includegraphics[width=\linewidth]{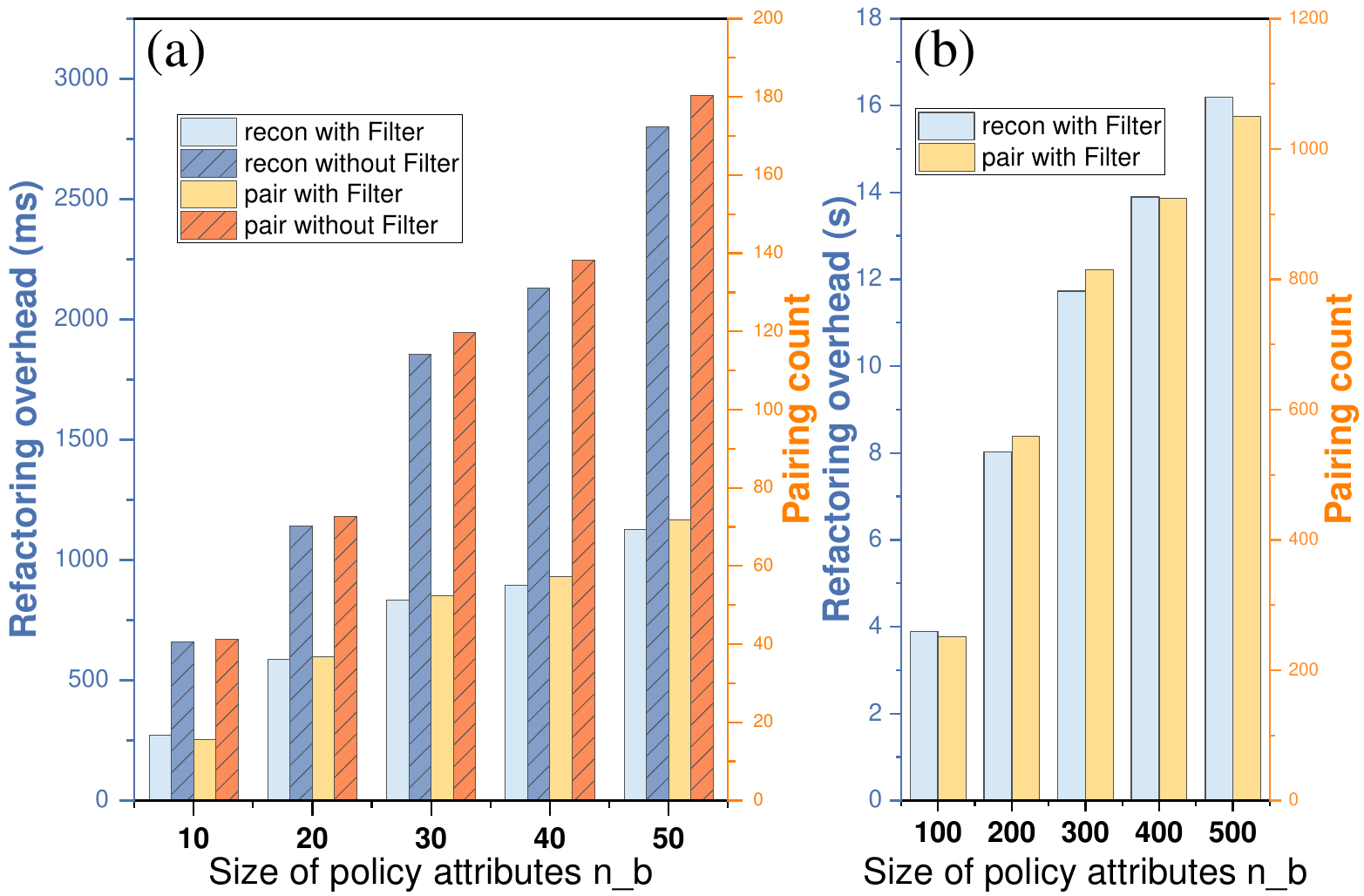}
		\caption{Impact of policy size $n_b$.
			(a) Small scale $(m{=}5, t_b{=}3)$, filter on vs.\ off.
			(b) Large scale $(m{=}20, t_b{=}5)$, filter on only.}
		\label{fig:buyer-nb}
	\end{figure}
	
	\subsubsection{Impact of buyer’s attribute count $m$}
	\label{sec:m-impact}
	
	Figure~\ref{fig:buyer-m}(a) varies $m$ from $3$ to $10$ with $n_b{=}20$, $t_b{=}3$.
	Without the filter, time grows linearly with $m$ (e.g., $2.6$\,s at $m{=}10$), because
	every additional attribute adds $n_b$ pairing attempts.  With the filter, the time stays
	almost flat (around $380$–$640$\,ms), as the effective candidate set remains $t_b{=}3$.
	Figure~\ref{fig:buyer-m}(b) scales $m$ up to $50$ with $n_b{=}200$, $t_b{=}5$ and filter
	enabled.  The time hovers around $5.8$–$8.6$\,s with no upward trend, further confirming
	that the filter successfully decouples the buyer’s workload from the size of its own
	attribute set.  Without the filter, the cost for $n_b{=}200$ would be on the order of
	$m \cdot n_b$ pairings, making large $m$ prohibitive.
	
	\begin{figure}[htbp]
		\centering
		\includegraphics[width=\linewidth]{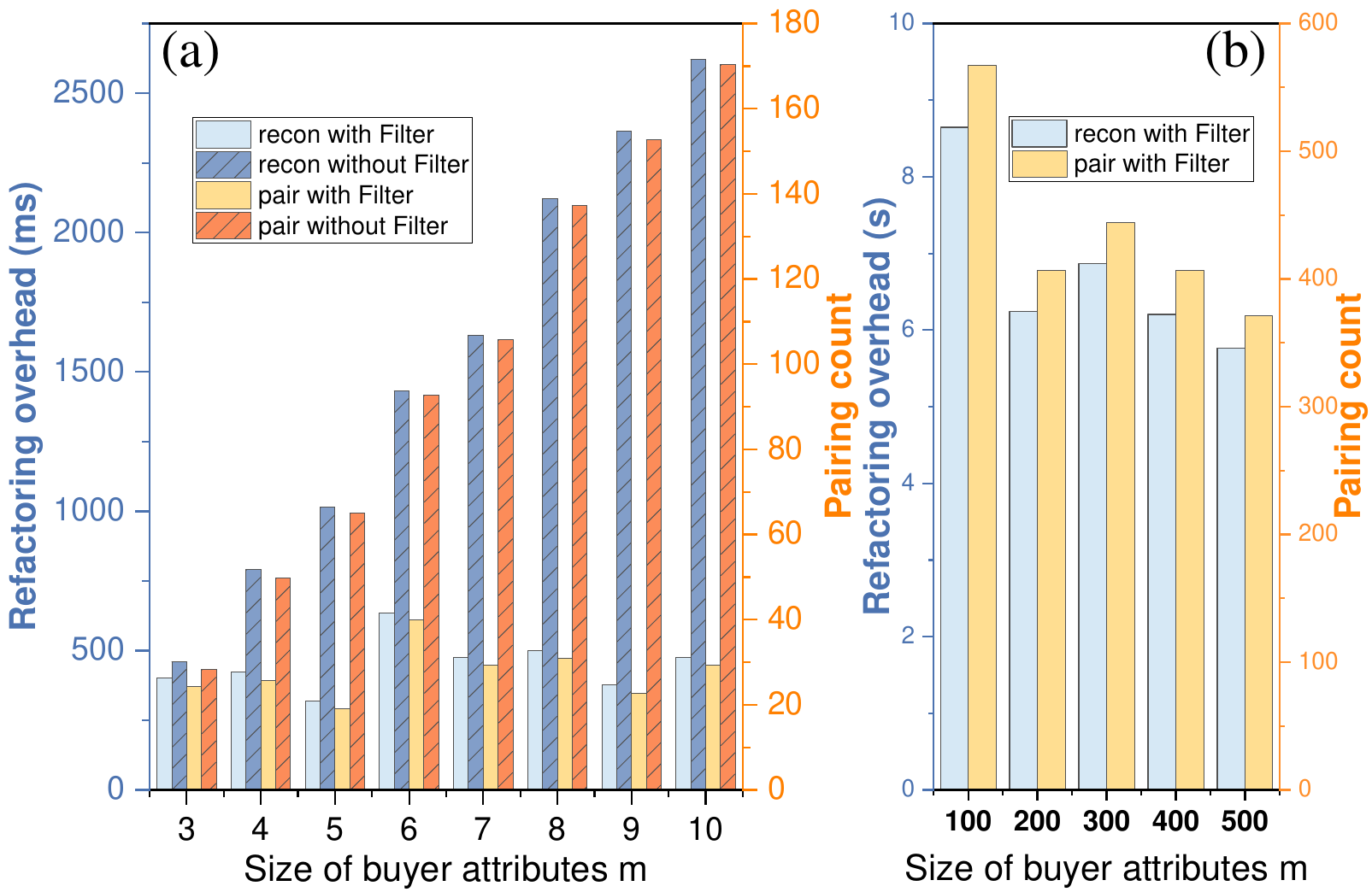}
		\caption{Impact of buyer attribute count $m$.
			(a) Small scale $(n_b{=}20, t_b{=}3)$, filter on vs.\ off.
			(b) Large scale $(n_b{=}200, t_b{=}5)$, filter on only.}
		\label{fig:buyer-m}
	\end{figure}
	
	\subsubsection{Impact of threshold $t_b$}
	\label{sec:tb-impact}
	
	Figure~\ref{fig:threshold}(a) studies $(n_b{=}20, m{=}5)$.  With the filter enabled,
	time increases from $288$\,ms ($t_b{=}2$) to $810$\,ms ($t_b{=}5$), consistent with the
	$t_b \cdot n_b$ pairing count.  Without the filter, the trend reverses: time drops from
	$1314$\,ms to $824$\,ms as $t_b$ grows.  This counter‑intuitive behaviour occurs because,
	for a small $n_b$, the exhaustive scans over non‑intersecting attributes dominate; when
	$t_b$ increases, fewer attributes lie outside the intersection, reducing the number of
	full scans.
	
	Figure~\ref{fig:threshold}(b) moves to the large scale $(n_b{=}200, m{=}20)$.  The
	filtered time rises linearly from $3.7$\,s ($t_b{=}2$) to $15.8$\,s ($t_b{=}10$).
	The unfiltered time stays approximately constant at $51.69 \pm 4.39$\,s (average over
	$t_b=2,\dots,10$), reflecting the dominant $m \cdot n_b$ term that dwarfs any $t_b$
	dependence.  Because the unfiltered cost is both huge and flat, we omit it from the plot
	and report only the filtered curve for clarity.
	
	In all scenarios, the filter not only accelerates reconstruction but also makes its
	cost predictable and dependent solely on $t_b$ (and $n_b$), eliminating the hidden
	overhead from large $m$.  This property is essential for practical deployment.
	
	\begin{figure}[htbp]
		\centering
		\includegraphics[width=\linewidth]{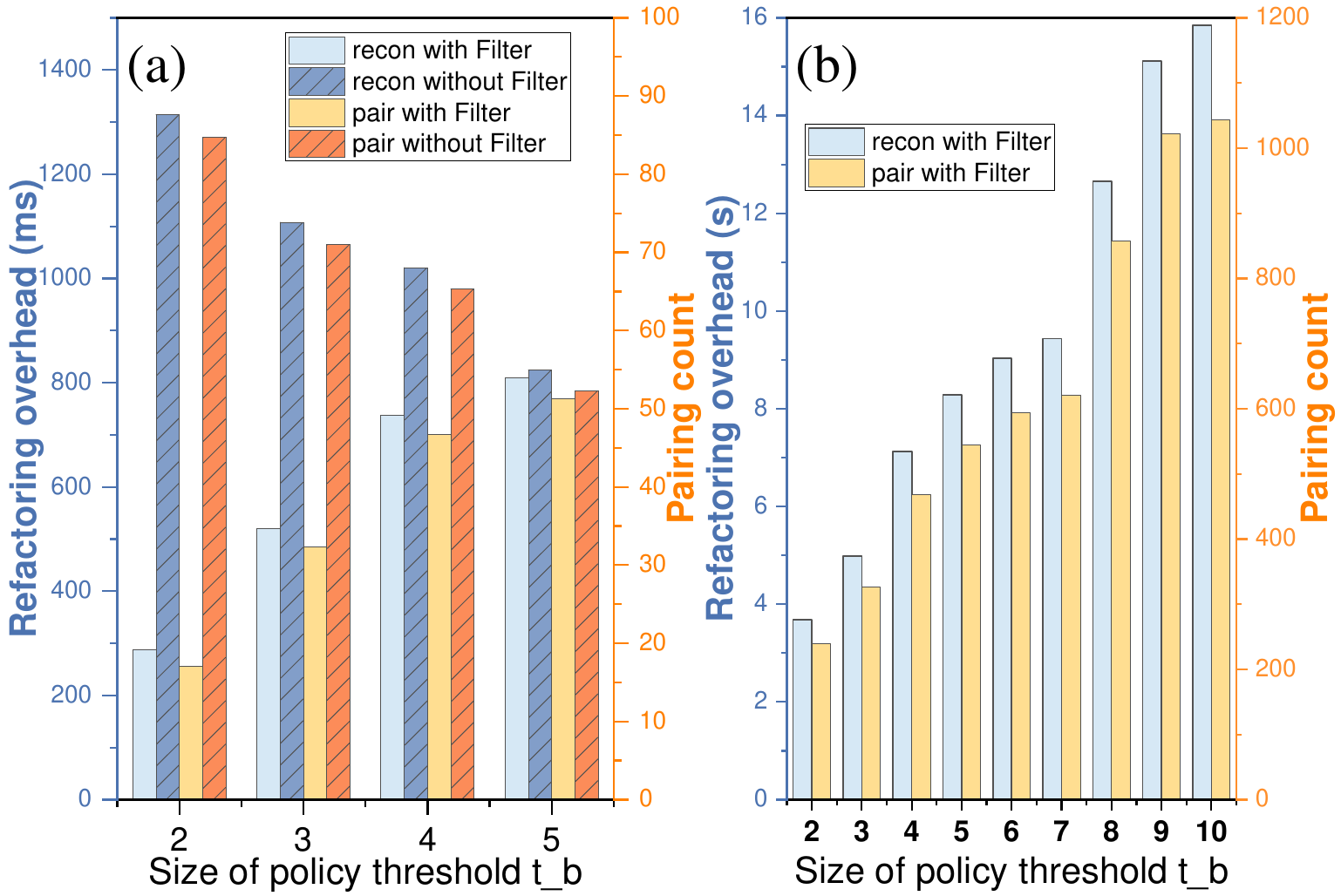}
		\caption{Impact of threshold $t_b$.
			(a) Small scale $(n_b{=}20, m{=}5)$, filter on vs.\ off.
			(b) Large scale $(n_b{=}200, m{=}20)$, filter on only.}
		\label{fig:threshold}
	\end{figure}
	
	\subsection{Comparison with Prior Work (RQ3)}
	\label{sec:comparison}
	
	We compare PriME‑Deal with three state‑of‑the‑art schemes that support
	attribute‑based matchmaking or fuzzy encryption: the basic IB‑ME of
	Ateniese~\emph{et~al.}~\cite{AtenieseFNV21}, the policy‑hiding PSME of
	Sun~\emph{et~al.}~\cite{Sun23PSME}, and the fuzzy IB‑ME of
	Wu~\emph{et~al.}~\cite{Wu2023Fuzzy}, the closest to our work because it
	also supports threshold matching.  We first provide an asymptotic analysis
	of the core off‑chain operations, then present a detailed experimental
	comparison with Fuzzy IB‑ME~\cite{Wu2023Fuzzy}.
	
	\subsubsection{Asymptotic complexity analysis}
	
	Table~\ref{tab:asymptotic} breaks down the dominant cryptographic operations
	for one execution of the seller‑side (Publish/Encrypt) and buyer‑side
	(Reconstruct/Decrypt) algorithms.  We count pairings ($P$), exponentiations
	in $\mathbb{G}$ or $\mathbb{G}_1$ ($E, E_1$), and hash‑to‑curve
	operations ($H$).  Constant terms are omitted for readability.
	
	The basic IB‑ME~\cite{AtenieseFNV21} and PSME~\cite{Sun23PSME} both incur
	$O(n)$ exponentiations for encryption and $O(1)$ pairings for decryption,
	but they do not support threshold matching.  Fuzzy IB‑ME introduces
	threshold matching via Shamir secret sharing, which raises the encryption
	cost to $O(n^2)$ pairings and $O(n^2)$ exponentiations (because of the
	polynomial evaluations of $T(x)$ and $H(x)$ for every attribute, each
	requiring $O(n)$ field operations and multiple pairings/exponentiations
	per polynomial term).  Decryption in Fuzzy IB‑ME is $O(t^2)$ pairings,
	as it must enumerate and combine at least $t$ shares using Lagrange
	coefficients.
	
	PriME‑Deal achieves the same threshold functionality with $O(n_b)$ pairings
	for the seller and $O(t_b\cdot n_b)$ pairings for the buyer.  Crucially,
	the buyer’s cost is independent of $m$ thanks to the Cuckoo filter,
	and neither side requires $O(n^2)$ operations.  Moreover, PriME‑Deal is the
	only scheme that provides bilateral attribute privacy and an auditable
	fair‑exchange mechanism.
	
	\begin{table}[htbp]
		\centering
		\caption{Asymptotic off‑chain complexity of matchmaking schemes.}
		\label{tab:asymptotic}
		\begin{tabular}{l c c}
			\toprule
			\textbf{Scheme} & \textbf{Publish / Encrypt} & \textbf{Reconstruct / Decrypt}\\
			\midrule
			IB‑ME~\cite{AtenieseFNV21}  & $n\,E + n\,E_1$ & $2P + 1E$\\
			PSME~\cite{Sun23PSME}       & $n\,E + n\,H$   & $2P + 1E$\\
			Fuzzy IB‑ME~\cite{Wu2023Fuzzy} & $n^2 P + n^2 E + n^2 E_1$ & $t^2 P$\\
			\textbf{PriME‑Deal}         & $n_b\,P + n_b\,E$ & $t_b\cdot n_b\,P$\\
			\bottomrule
		\end{tabular}
	\end{table}
	
	\subsubsection{Experimental comparison with Fuzzy IB‑ME}
	
	\begin{figure}[htbp]
		\centering
		\includegraphics[width=\linewidth]{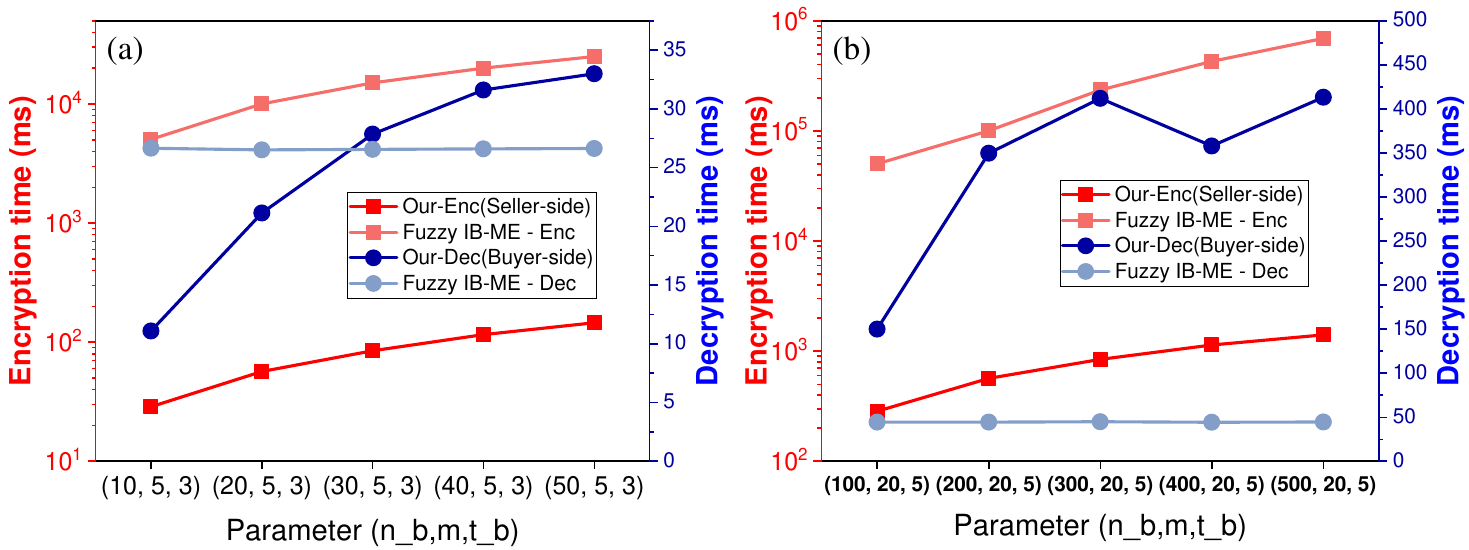}
		\caption{Comparison with Fuzzy IB‑ME.
			(a) Small scale $(m{=}5, t_b{=}3)$.
			(b) Large scale $(m{=}20, t_b{=}5)$.}
		\label{fig:comp-nb}
	\end{figure}
	
	Because Fuzzy IB‑ME is the only prior scheme that supports threshold
	matching, we focus the experimental comparison on it.  For a fair
	comparison, both schemes are benchmarked on the same SS512 symmetric curve
	(112‑bit security).  Figure~\ref{fig:comp-nb}(a) plots the publish
	(encryption) time for small policy sizes ($n_b = 10\text{--}50$, $m=5$,
	$t_b=3$), and Figure~\ref{fig:comp-nb}(b) extends the scale to $n_b =
	100\text{--}500$ ($m=20$, $t_b=5$).
	
	Fuzzy IB‑ME’s encryption grows super‑linearly and exceeds $690$\,s at
	$n_b=500$, whereas PriME‑Deal’s publish remains under $1.5$\,s.  This
	two‑orders‑of‑magnitude gap originates from the $O(n^2)$ polynomial
	evaluations in Fuzzy IB‑ME: for every attribute in the policy and every
	attribute in the seller’s set, the algorithm must evaluate the polynomials
	$T(x)$ and $H(x)$, each of degree $n$, requiring $O(n)$ exponentiations
	per evaluation, leading to $O(n^2)$ pairings and exponentiations overall.
	In contrast, PriME‑Deal’s seller only computes one pairing and one
	exponentiation per policy attribute, yielding linear $O(n_b)$ complexity.
	
	On the buyer side, Fuzzy IB‑ME achieves a near‑constant decryption time
	(${\approx}44$\,ms) independent of $n_b$, $m$, and $t_b$, because the
	threshold matching is algebraically embedded in the ciphertext structure:
	the decryptor directly combines $t$ pre‑computed components without
	iterating over the full attribute set.  PriME‑Deal’s reconstruction time
	grows linearly with $n_b$ (from $11$\,ms to $413$\,ms across the two
	scales), as the buyer must pair each of its $t_b$ candidate attributes
	with up to $n_b$ ciphertexts to locate the correct tags.  However, this
	linear growth is predictable and remains well within interactive bounds
	even for $n_b=500$.  Moreover, PriME‑Deal’s buyer does not require any
	external hint about which subset of attributes to use---the tag‑based
	matching automatically identifies the correct shares---and the protocol
	additionally hides both the seller’s and the buyer’s attribute sets,
	which Fuzzy IB‑ME does not.

	\noindent
	In summary, PriME‑Deal achieves the best trade‑off between off‑chain
	efficiency and functionality: it is orders of magnitude faster than Fuzzy
	IB‑ME on the seller side, keeps the buyer online time within practical
	bounds, and is the only scheme that simultaneously offers threshold
	matching, bilateral attribute privacy, and auditable fair exchange.

	\section{Conclusion and Future Work}
	
	We presented PriME-Deal, a non‑interactive protocol that achieves
	policy‑hiding bilateral attribute matching, linear‑time access control, and
	auditable fair exchange on a public blockchain.  The core idea is to decouple
	attribute matching from file decryption by inserting a secret token: the
	seller publishes masked shares in a single oblivious key‑value store, and the
	buyer reconstructs the token entirely off‑line, then proves compliance in
	zero‑knowledge.  Fairness is enforced through a collateralized on‑chain reveal
	and a Merkle‑based audit that penalizes misbehavior without trusted parties.
	
	Experiments on BN254 show that the seller’s publish time scales linearly
	($3.5$\,s for $n_b{=}200$, $8.76$\,s for $n_b{=}500$), achieving a
	\textbf{two orders of magnitude speedup} over the state‑of‑the‑art threshold
	fuzzy IB‑ME (which requires $690$\,s for $500$ attributes).  The buyer’s
	online latency is under $9$\,s for typical configurations, the Groth16 proof
	generation is constant at $0.6$\,s, and the Cuckoo filter makes the buyer’s
	workload independent of her attribute set size.  On‑chain verification
	consumes $28.6$\,M gas, well within Ethereum’s block limit.
	
	\textbf{Future Work.}
	Several directions merit further investigation. First, for very large policies ($n_b > 500$), integrating list‑decodable Shamir schemes or batched pairing techniques could reduce the buyer's reconstruction latency.  Second, the current trusted CA can be decentralized via threshold BLS or multi‑authority ABE to strengthen the trust model.  Third, migrating the pairing‑based components to post‑quantum assumptions would extend the protocol's applicability to long‑term security.  Finally, optimizing the ZKP circuit for larger thresholds ($t_b > 5$) would enable more fine‑grained access policies without sacrificing constant proof size.

	\clearpage
	
	\bibliographystyle{IEEEtran}
	\bibliography{IEEEabrv,ref}

	\appendices
	
		\section{Compact Protocol Specification of PriME‑Deal}
	\label{app:PriME‑Deal}
	
	Figure~\ref{fig:prime-deal-full} shows the compact protocol specification of PriME-Deal.
	
	\begin{figure*}
		\centering
		\makebox[\textwidth]{
			\fbox{
				\parbox{\dimexpr\textwidth-2\fboxsep-2\fboxrule}
				{
					\underline{\textit{Phase 0: System Setup \& Key Generation (CA, off-chain)}}
					\begin{itemize}
						\item[0a] Generate bilinear group, $\mathsf{mpk}=g^s$ ($\mathsf{msk}=s$), RO hashes $H_1..H_4$, PRF, OKVS row function, and Groth16 CRS.
						\item[0b] Seller: $\mathsf{sk}_S = s\!\cdot\!H_3(\mathsf{uid}_S)$, $\mathsf{pk}_S=g^{\mathsf{sk}_S}$, $\mathsf{cred}_{\mathcal{A}_S}=H_1(\mathsf{uid}_S\|\mathcal{A}_S)^s$.
						\item[0c] Buyer: $\mathsf{sk}_B = s\!\cdot\!H_3(\mathsf{uid}_B)$, $\mathsf{pk}_B=g^{\mathsf{sk}_B}$, $\mathsf{cred}_{\mathcal{A}_B}=H_1(\mathsf{uid}_B\|\mathcal{A}_B)^s$, and $\mathsf{DK}_B=\{D_a=H_1(a)^s\}_{a\in\mathcal{A}_B}$.
					\end{itemize}
					
					\underline{\textit{Phase 1: Publish (Seller $\mathcal{S}$, off-chain)}}
					\begin{itemize}
						\item[(1)] Pick random token $\tau$, $\rho$, $k_2$; set $C_0 = g^\rho$,
						$k_1 = H_2(e(H_1(\tau), C_0))$, $k = k_1 \oplus k_2$.
						\item[(2)] Split $F$ into $m$ blocks. For $i=1..m$: $\kappa_i = H_4(k\|i)$, $c_i = \mathsf{AES.Enc}(\kappa_i, F_i)$, $h_i = H_4(c_i)$. 
						$C_F = c_1\|\dots\|c_m$. Build Merkle tree on $\{h_i\}$, obtain root $\mathsf{root}$, sign $\sigma_{\mathsf{root}} = H_1(\mathsf{root})^{\mathsf{sk}_S}$.
						\item[(3)] $(t_b,n_b)$-Shamir sharing of $\tau$ over $\mathcal{P}_B$: choose $f(x)$ of degree $t_b{-}1$, $f(0)=\tau$; $x_i = H_3(p_i)$, $s_i = f(x_i)$.
						\item[(4)] Double-mask shares: pick seeds $\mathsf{s}_\nu,\mathsf{s}_r$; for each $p_i$, $r_i = \mathsf{PRF}(\mathsf{s}_r, p_i)$, $C_i^* = g^{r_i}$, $K_i = H_2(e(H_1(p_i), \mathsf{mpk})^{r_i})$, $\mathsf{tag}_i = H_4(K_i)$, $\nu_i = \mathsf{PRF}(\mathsf{s}_\nu, p_i)$, $y_i = s_i + \nu_i + H_3(K_i) \bmod p$.
						\item[(5)] Encode into OKVS: $D \gets \mathsf{OKVS.Encode}(\{(p_i, y_i)\})$, and publish ciphertext set $\mathcal{C} = \{(C_i^*,\mathsf{tag}_i)\}_{i=1}^{n_b}$.
						\item[(6)] Build Cuckoo filter $\mathsf{CF}_{\mathsf{hint}}$ containing $\{H_3(p_i)\}$.
						\item[(7)] Sign: $S_0 = H_1(\mathsf{sid}\|\tau)^{\mathsf{sk}_S}$,
						$S_1 = H_1(\mathsf{info})^{\mathsf{sk}_S}$.
						\item[(8)] Publish $\mathsf{aux} = (\mathsf{info}, S_1)$, where
						\[
						\mathsf{info} = (\mathsf{sid}, \mathsf{IC}_S, C_F, \mathsf{s}_\nu, C_0, \mathcal{C}, S_0, D, \mathsf{CF}_{\mathsf{hint}}, \mathsf{com}_{\mathcal{P}_B}, t_b),
						\]
						$\mathsf{com}_{\mathcal{P}_B} = H_4(\mathcal{P}_B \| H_4(\tau))$.
						\item[(9)] Deposit to $\mathcal{SC}$:
						$(\mathsf{sid}, h_D=H_4(D), \mathsf{com}_{\mathcal{P}_B}, t_b, S_0, h_{k_2}=H_4(k_2), \mathsf{root}, \sigma_{\mathsf{root}}, v, d)$,
						with collateral $d=\lceil\log{n_b}\rceil\cdot v$.
					\end{itemize}
					
					\underline{\textit{Phase 2: Bilateral Verification, Token Recovery \& Proof (Buyer $\mathcal{B}$, off-chain)}}
					\begin{itemize}
						\item[(10)] Retrieve $\mathsf{aux}$; verify $e(S_1,g)\stackrel{?}{=}e(H_1(\mathsf{info}),\mathsf{pk}_S)$, verify $\mathsf{cred}_{\mathcal{A}_S}$, and check $|\mathcal{P}_S\cap\mathcal{A}_S|\ge t_s$. Abort if any fails.
						\item[(11)] Pre-filter attributes: $\mathcal{A}' = \{a\in\mathcal{A}_B \mid \mathsf{CF}_{\mathsf{hint}}.\mathsf{lookup}(H_3(a))=\text{true}\}$.
						\item[(12)] For each $a\in\mathcal{A}'$: $z_a = \langle\mathsf{row}(a), D\rangle$, $\nu_a = \mathsf{PRF}(\mathsf{s}_\nu, a)$.
						For each $(C_i^*,\mathsf{tag}_i)\in\mathcal{C}$: compute $K_i' = H_2(e(C_i^*, D_a))$,
						\textbf{if $H_4(K_i') = \mathsf{tag}_i$}, collect $u = z_a - \nu_a - H_3(K_i') \bmod p$.
						\item[(13)] Reconstruct token: enumerate all $t_b$-subsets of collected $(H_3(a), u)$ pairs; use Lagrange interpolation to obtain $\tau'$. Abort if none valid.
						\item[(14)] Authenticate token: check $e(S_0,g) \stackrel{?}{=} e(H_1(\mathsf{sid}\|\tau'), \mathsf{pk}_S)$.
						\item[(15)] Generate zk-SNARK $\pi$ for $\mathcal{R}_{\mathsf{buy}}$ (Alg.~\ref{alg:buyer-circuit}).\\
						\emph{Public input} $x = (\mathsf{sid}, \mathsf{com}_{\mathcal{A}_B}, h_D, \mathsf{s}_\nu, t_b, S_0, C_0, \mathsf{mpk}, \mathsf{pk}_S, \mathsf{uid}_S)$;\\
						\emph{Witness} $w$ contains $(\mathcal{A}_B, \mathsf{cred}_{\mathcal{A}_B}, \mathsf{sk}_{\mathsf{uid}_B}, \tau', \mathcal{U}_0, D)$.
						\item[(16)] Submit $\mathsf{tx} = (\text{``buy''}, \mathsf{sid}, \mathsf{com}_{\mathcal{A}_B}, \pi, v)$ to $\mathcal{SC}$.
					\end{itemize}
					
					\underline{\textit{Phase 3: On-Chain Optimistic Fair Exchange (Smart Contract)}}
					\begin{itemize}
						\item[(17)] $\mathcal{SC}$ verifies $\pi$ using public inputs from its storage. If valid, freeze buyer's $v$ and start challenge period $\Delta T$.
						\item[(18)] \textbf{During} $\Delta T$, seller may call $\texttt{reveal}(\mathsf{sid}, k_2')$:
						if $H_4(k_2')=h_{k_2}$, store $k_2'$, transfer $v$ to seller, return $d$ to seller, emit $\mathsf{KeyReleased}$, and start dispute window $\Delta T_{\mathsf{disp}}$.\\
						\noindent \textbf{Audit} \textbf{During} $\Delta T_{\mathsf{disp}}$, if decryption fails, buyer calls $\texttt{challenge}$:
						\begin{itemize}
							\item $\mathcal{SC}$ selects a random set $\mathcal{I}$ of block indices (Sec.~\ref{sec:audit-evidence}).
							\item Buyer submits $\{c_i, b_i, \pi_i\}_{i\in\mathcal{I}}$, where $\pi_i$ is a Merkle proof for $h_i=H_4(c_i)$.
							\item $\mathcal{SC}$ verifies: (i) each Merkle proof against $\mathsf{root}$ and $\sigma_{\mathsf{root}}$; (ii) $\mathsf{AES.Enc}(H_4(k\|i), b_i) = c_i$.
							\item If all checks pass, seller is honest, seller gets $v$, collateral returned. Else, seller cheated: buyer gets $v+d$, listing closed.
						\end{itemize}
						\item[(19)] If $\Delta T$ expires with no valid $\texttt{reveal}$, buyer calls $\texttt{reclaim}()$ to get back $v$ and receives seller's collateral $d$.
					\end{itemize}
					
					\underline{\textit{Phase 4: File Decryption \& Audit Trigger (Buyer $\mathcal{B}$, off-chain)}}
					\begin{itemize}
						\item[(20)] Upon $\mathsf{KeyReleased}$, get $k_2$; compute $k_1' = H_2(e(H_1(\tau'), C_0))$, recover file key $k = k_1' \oplus k_2$.
						\item[(21)] Decrypt: $F \gets \mathsf{AES.Dec}(k, C_F)$. Authenticated encryption guarantees integrity.
						\item[(21b)] If decryption fails, prepare audit evidence as per Sec.~\ref{sec:audit-evidence} and call $\texttt{challenge}$ (Phase~3 (18)).
					\end{itemize}
					
					\vspace{0.3cm}
					\noindent\textbf{Correctness and Security Highlights.}
					For $a \in \mathcal{P}_B \cap \mathcal{A}_B$, the OKVS decode plus pairing tag match yields a valid Shamir share; any $t_b$ overlapping attributes reconstruct $\tau$. 
					The zero-knowledge proof enforces that the buyer indeed owns a certified attribute set satisfying the policy, preventing collusion.
					Seller-side policy is enforced locally (step~10).
					The optimistic fair exchange with collateral $d>v$ ensures a rational seller reveals correct $k_2$, while the Merkle-based public audit) guarantees that any cheating seller is detected with overwhelming probability and the buyer is compensated. Thus, bilateral access control and fair exchange are simultaneously achieved.
				}
			}
		}
		\caption{Compact and precise specification of PriME‑Deal (aligned with Algorithms~\ref{alg:publish}--\ref{alg:buyer-circuit}).}
		\label{fig:prime-deal-full}
	\end{figure*}

	\section{Detailed False‑Positive Analysis}
	\label{app:fp}
	
	To systematically evaluate the impact of the Bloom filter's false‑positive rate
	$\varepsilon$ on reconstruction performance and correctness, we inject $200$
	non‑policy attributes into the buyer's attribute set and measure the behaviour
	under three different scales.  Each configuration is repeated $10$ times;
	Table~\ref{tab:fp-full} reports the average and standard deviation of the
	reconstruction time and the number of pairing operations, while the
	false‑positive count (FP) is listed as an exact value because it exhibited
	zero variance across runs.  All $120$ trials succeeded in recovering the
	correct token.
	
	\begin{table}[htbp]
		\centering
		\caption{False‑positive experiment (10 runs, mean $\pm$ std).  Full data including 0.5\% error rate.}
		\label{tab:fp-full}
		\resizebox{\linewidth}{!}{%
			\begin{tabular}{c c c c c}
				\toprule
				\makecell[c]{\textbf{Scale}\\ $(n_b,m,t_b)$} & $\varepsilon$ & \textbf{FP} & \textbf{Time (ms)} & \textbf{Pairings} \\
				\midrule
				\multirow{4}{*}{\makecell[c]{Small \\ (50,5,3)}}
				& 0.1\% & $0.0$ & $1199.6\pm296.1$   & $76.8\pm19.4$ \\
				& 0.5\% & $0.0$ & $1183.8\pm323.5$   & $75.8\pm21.3$ \\
				& 1.0\% & $0.0$ & $1332.3\pm325.6$   & $85.5\pm21.4$ \\
				& 5.0\% & $3.0$ & $3451.6\pm417.2$   & $225.0\pm27.6$ \\
				\midrule
				\multirow{4}{*}{\makecell[c]{Medium \\ (200,20,5)}}
				& 0.1\% & $0.0$ & $8845.3\pm1973.5$  & $579.6\pm129.4$ \\
				& 0.5\% & $0.0$ & $7672.3\pm3064.5$  & $502.8\pm202.0$ \\
				& 1.0\% & $0.0$ & $7258.7\pm1674.2$  & $475.8\pm110.5$ \\
				& 5.0\% & $1.0$ & $10084.3\pm2256.3$ & $661.6\pm148.6$ \\
				\midrule
				\multirow{4}{*}{\makecell[c]{Large \\ (500,50,5)}}
				& 0.1\% & $0.0$ & $21255.2\pm4832.9$ & $1395.3\pm318.3$ \\
				& 0.5\% & $0.0$ & $16986.1\pm6977.2$ & $1114.7\pm459.0$ \\
				& 1.0\% & $0.0$ & $18903.6\pm2433.3$ & $1240.0\pm160.1$ \\
				& 5.0\% & $2.0$ & $33059.3\pm5261.4$ & $2171.7\pm345.4$ \\
				\bottomrule
			\end{tabular}%
		}
	\end{table}
	
	\noindent
	Across all scales, at the design target $\varepsilon = 0.1\%$ the observed
	false‑positive count is consistently \textbf{zero}.  Even when
	$\varepsilon$ is deliberately increased to $5\%$, no configuration ever
	suffered a token recovery failure, because false‑positive attributes lack a
	valid tag $H_4(K_i)$ and are therefore discarded before the Lagrange
	interpolation step.  The occasional false positives merely add a modest
	number of extra pairings (roughly $n_b$ per spurious attribute), which
	explains the slightly higher reconstruction times at larger $\varepsilon$.
	This confirms that the Bloom filter can be safely used with the recommended
	$0.1\%$ setting.

	\section{Notation and Full Security Proofs}
	\label{app:full-proofs}

	\subsection{Notation}
	\label{app:table}

	\small
	\renewcommand{\arraystretch}{0.9}
	\noindent
	\begin{tabularx}{0.48\textwidth}{@{} l X @{}}
		\toprule
		\textbf{Symbol} & \textbf{Meaning} \\
		\midrule
		$\mathcal{P}_B, n_b, t_b$ & seller's buyer-side policy, its size, and threshold \\
		$\mathcal{P}_S, t_s$ & buyer's seller-side policy and threshold \\
		$\mathcal{A}_B, \mathcal{A}_S$ & buyer's and seller's attribute sets \\
		$m$ & Number of plaintext blocks / buyer's attribute count (context‑dependent) \\
		$\tau, \rho$ & secret token and randomness for $C_0 = g^\rho$ \\
		$k = k_1 \oplus k_2$ & file key; $k_1$ from token, $k_2$ escrowed on-chain \\
		$\kappa_i$ & block-specific key for file block $i$ (to avoid confusion with $k_1,k_2$) \\
		$c_i, h_i$ & Encrypted block $i$ and its hash \\
		$f(x), s_i$ & Shamir polynomial and share for attribute $p_i$ \\
		$C_i = (C_i^*, \mathsf{tag}_i)$ & ciphertext for attribute $p_i$, with pairing component $C_i^* \in \mathbb{G}$ and match tag $\mathsf{tag}_i$ \\
		$\mathcal{C}$ & Set of ciphertexts $\{C_1, \dots, C_{n_b}\}$ \\
		$D$, $D_a$ & OKVS vector and its attribute decryption key for attribute $a$ \\
		$\mathsf{s}_\nu, \mathsf{s}_r$ & PRF keys for masking and encryption randomness \\
		$\mathsf{CF}_{\mathsf{hint}}$ & Cuckoo filter over $\{H_3(p_i)\}_{i=1}^{n_b}$ \\
		$\mathsf{com}_{\mathcal{A}_B}$ & Commitment to buyer attribute set \\
		$C_F, \mathsf{root}$ & Encrypted file blocks concatenation and Merkle tree root \\
		$v, d$ & payment and collateral ($d = \lceil\log n_b\rceil \cdot v$) \\
		\bottomrule
	\end{tabularx}
	\normalsize
	
	\subsection{UC‑Security – Complete Simulation}
	\label{app:uc}
	
	\begin{theorem}[UC‑Security, detailed]
		\label{thm:uc-weak-full}
		Assume the OKVS is oblivious and its composition with a PRF achieves
		value‑hiding, the signature scheme is existentially unforgeable under
		co‑CDH, the key derivation $H_2(e(H_1(\tau), g^\rho))$ is pseudorandom
		under CBDH, and the zk‑SNARK is zero‑knowledge and knowledge‑sound.
		Then the PriME‑Deal protocol $\Pi_{\mathsf{PD}}$ UC‑realizes
		$\mathcal{F}_{\mathsf{PPFE}}^{\mathsf{weak}}$ in the
		$\mathcal{F}_{\mathsf{ledger}}$-hybrid model against static corruptions
		(the CA is never corrupted).
	\end{theorem}
	
	\begin{proof}
		We construct a simulator $\mathcal{S}$ that internally runs the real‑world
		adversary $\mathcal{A}$ and interacts with the ideal functionality
		$\mathcal{F}_{\mathsf{PPFE}}^{\mathsf{weak}}$.
		
		\medskip\noindent\textbf{Setup.}
		$\mathcal{S}$ generates the public parameters $\mathsf{pp}$ as in the real
		$\mathsf{Setup}$, except that it obtains the common reference string
		$\mathsf{crs}$ from the ZKP simulator.  By zero‑knowledge, the simulated
		$\mathsf{crs}$ is indistinguishable from a real one.
		
		\medskip\noindent\textbf{Key registration.}
		The CA is honest.  For every registration request, $\mathcal{S}$ executes
		the real $\mathsf{AttrKeyGen}$ using the master secret $s$ (which it knows)
		and returns $\mathsf{IC}$ and, for buyers, $\mathsf{DK}$.  The view is
		identical to the real world.
		
		\medskip\noindent\textbf{Honest seller publish.}
		On input $(\mathsf{Publish}, \mathsf{sid}, \mathcal{A}_S, \mathcal{P}_B, t_b,
		F, v, d)$ from the environment, $\mathcal{S}$ sends it to
		$\mathcal{F}_{\mathsf{PPFE}}^{\mathsf{weak}}$.  The functionality generates
		a random $\tau \in \mathbb{F}_p$ and $K \in \{0,1\}^\lambda$, stores the
		listing, and returns $(\mathsf{sid}, \mathsf{seller}, |\mathcal{A}_S|,
		|\mathcal{P}_B|, t_b, v, d)$.  $\mathcal{S}$ must now produce
		$\mathsf{aux}$ and the on‑chain deposit.
		
		\begin{itemize}[leftmargin=10px]
			\item \textbf{Simulating the OKVS.}
			$\mathcal{S}$ picks uniformly random $\tilde{y}_1,\dots,\tilde{y}_{n_b}
			\in \mathbb{F}_p$ and dummy keys $\tilde{p}_i$ unrelated to
			$\mathcal{P}_B$, and sets $D \gets \mathsf{Encode}(\{(\tilde{p}_i,
			\tilde{y}_i)\})$.  By the value‑hiding property of the OKVS+PRF
			composition, this $D$ is indistinguishable from a real encoding.
			Because the adversary is not allowed to corrupt a set of buyers that
			jointly satisfy the threshold $t_b$ (otherwise it could trivially
			distinguish the two policies, contradicting the UC framework's
			corruption rules), it cannot detect the substitution.
			
			\item \textbf{Other fields.}
			$\mathcal{S}$ picks random $\rho \in \mathbb{F}_p$ and $k_2 \in
			\{0,1\}^\lambda$, sets $C_0 = g^\rho$, $k_1 = H_2(e(H_1(\tau), C_0))$,
			$k = k_1 \oplus k_2$, encrypts $C_F \gets \mathsf{AES.Enc}(k, F)$,
			computes $S_0 = H_1(\mathsf{sid}\|\tau)^s$ and $S_1$ on $\mathsf{info}$
			exactly as in Algorithm~\ref{alg:publish}, and generates $\mathcal{C}$
			as random group elements (their real distribution depends only on PRF
			outputs and is thus pseudorandom).  It then deposits
			$(\mathsf{sid}, H_4(D), \mathsf{com}_{\mathcal{P}_B}, t_b, S_0,
			h_{k_2}, v, d)$ into the simulated contract.
		\end{itemize}
		The only differences from a real $\mathsf{aux}$ are in $D$ and
		$\mathcal{C}$, both of which are computationally indistinguishable
		under the stated assumptions.
		
		\medskip\noindent\textbf{Honest buyer operation.}
		The ideal functionality $\mathcal{F}_{\mathsf{PPFE}}^{\mathsf{weak}}$ does
		\emph{not} check the seller‑side policy; this check is performed locally
		by an honest buyer before calling $\mathsf{Buy}$.  Consequently, the
		environment will only instruct an honest buyer to execute $\mathsf{Buy}$
		when $|\mathcal{P}_S \cap \mathcal{A}_S| \ge t_s$ holds.  Hence we only
		need to consider the buyer‑side policy and the ZK proof.
		
		When an honest buyer is activated with $(\mathsf{Buy}, \mathsf{sid},
		\mathcal{P}_S, t_s, \mathcal{A}_B, \pi, v)$, $\mathcal{S}$ forwards it to
		$\mathcal{F}_{\mathsf{PPFE}}^{\mathsf{weak}}$.  The functionality checks
		$|\mathcal{P}_B \cap \mathcal{A}_B| \ge t_b$ and the validity of $\pi$.
		If the check fails, it rejects; otherwise it accepts, freezes $v$, and
		starts the timer (without immediately delivering the file key).
		
		\begin{itemize}[leftmargin=10px]
			\item \textbf{Rejection.}
			If the functionality rejects, $\mathcal{S}$ runs the honest buyer
			algorithm on the previously generated $D$ (which contains only random
			values).  Because the OKVS decodes to random values independent of
			$\tau$, reconstruction fails with overwhelming probability and the
			buyer aborts.  This transcript is indistinguishable from a real failed
			attempt.
			
			\item \textbf{Acceptance.}
			If the functionality accepts, $\mathcal{S}$ constructs a valid witness
			$w$ containing the buyer's attributes, the credential, the reconstructed
			$\tau$ (which $\mathcal{S}$ knows from the publish phase), and a set of
			shares $\mathcal{U}_0$ that satisfy the circuit constraints.  It then
			generates a simulated proof $\pi$ via the ZKP simulator.  By
			zero‑knowledge, $\pi$ is indistinguishable from a real proof.
			$\mathcal{S}$ submits $(\text{``buy''}, \mathsf{sid},
			\mathsf{com}_{\mathcal{A}_B}, \pi, v)$ to the contract, which verifies
			$\pi$ and freezes the payment.
		\end{itemize}
		
		\medskip\noindent\textbf{Key reveal and timeout.}
		$\mathcal{S}$ monitors the real‑world contract.  If the seller invokes
		$\texttt{reveal}$ with $k_2'$ where $H_4(k_2') = h_{k_2}$, $\mathcal{S}$
		instructs $\mathcal{F}_{\mathsf{PPFE}}^{\mathsf{weak}}$ to deliver
		$(\mathsf{KeyReveal}, \mathsf{sid}, k_2')$; the functionality transfers
		$v$ to the seller and sends the derived file key to the buyer.  If no
		valid reveal occurs before the challenge timer expires, $\mathcal{S}$
		triggers $\mathsf{Timeout}$; the functionality returns $v$ and $d$ to
		the buyer.  In both cases the outcomes match the real contract.
		
		\medskip\noindent\textbf{Corrupted parties.}
		For a corrupted seller, $\mathcal{S}$ follows the real protocol with the
		keys provided by $\mathcal{A}$.  For a corrupted buyer, $\mathcal{S}$
		runs the real verification algorithm on the submitted $\pi$.  If $\pi$
		is valid, knowledge‑soundness yields an extractor that outputs a witness
		containing a certified attribute set $\mathcal{A}_B$ and a token $\tau'$
		satisfying $|\mathcal{P}_B \cap \mathcal{A}_B| \ge t_b$ and the seller's
		signature.  $\mathcal{S}$ extracts this witness and supplies it to the
		ideal functionality to pass the buyer‑side policy check.  If the proof
		is invalid, the transaction is rejected.
		
		A standard hybrid argument shows that the real and ideal views are
		computationally indistinguishable.  The only non‑trivial steps are the
		OKVS substitution (secured by value‑hiding) and the ZKP simulation
		(secured by zero‑knowledge).  Hence $\Pi_{\mathsf{PD}}$ UC‑realizes
		$\mathcal{F}_{\mathsf{PPFE}}^{\mathsf{weak}}$.
	\end{proof}
	
	\subsection{Full Proof of Policy Privacy}
	\label{app:priv}
	
	\begin{theorem}[Policy Privacy]
		\label{thm:priv-full}
		If the OKVS+PRF composition is value‑hiding, then
		$\mathsf{Priv}_{\mathcal{A}}^{\mathsf{policy}}(\lambda)$ is negligible
		for any PPT adversary $\mathcal{A}$.
	\end{theorem}
	
	\begin{proof}
		Consider the policy privacy game of Section~\ref{sec:games}.  We define two
		hybrid experiments:
		
		\begin{itemize}[leftmargin=10px]
			\item \textbf{Hybrid $H_0$:} This is the real game.  The challenger uses
			the bit $b$ chosen by $\mathcal{A}$ to generate $\mathsf{aux}$ and the
			on‑chain deposit following $\mathsf{Publish}$ exactly as in the real
			protocol.
			\item \textbf{Hybrid $H_1$:} The challenger replaces the OKVS vector $D$
			with an encoding of $n_b$ uniformly random field elements, independently
			of the token $\tau$ and the policy $\mathcal{P}_{B,b}$.  All other
			components ($\mathsf{CF}$, $S_0$, $C_0$, $\mathcal{C}$, commitments, etc.)
			are generated honestly using the same $\tau$ and $k_2$.
		\end{itemize}
		
		\noindent\textbf{Indistinguishability of $H_0$ and $H_1$.}
		The only difference between the two hybrids is the content of $D$.
		In $H_0$, $D$ encodes the PRF‑masked Shamir shares of $\tau$ under
		$\mathcal{P}_{B,b}$.  In $H_1$, $D$ encodes purely random values.
		The value‑hiding property of the OKVS+PRF composition (Definition~\ref{sec:okvs})
		guarantees that these two distributions are computationally indistinguishable
		for any adversary that does not possess a sufficiently large set of attributes
		to reconstruct $\tau$.  The policy privacy game explicitly forbids the
		adversary from corrupting parties or requesting keys that would enable it to
		satisfy the threshold $t_b$ for both challenge branches simultaneously;
		hence the adversary cannot distinguish the two distributions.  Any
		distinguisher between $H_0$ and $H_1$ can be turned into an adversary against
		value‑hiding, so the distinguishing advantage is negligible.
		Even though the PRF seed $\mathsf{s}_\nu$ is public, the adversary can
		compute $\nu_a = \mathsf{PRF}(\mathsf{s}_\nu, a)$ for any candidate $a$.
		The value stored in the OKVS is $y_i = s_i + \nu_i + H_3(K_i)$.
		Without the attribute decryption key $D_a$, the adversary cannot compute
		$H_3(K_i)$ because the pairing output $e(C_i, D_a)$ is computationally
		hidden under the CBDH assumption (the adversary does not possess $s$).
		Hence the mask $H_3(K_i)$ remains pseudorandom, and the value‑hiding
		property of the OKVS+PRF construction is preserved even with the knowledge
		of $\mathsf{s}_\nu$.
		
		\noindent\textbf{Advantage in $H_1$.}
		In $H_1$, the OKVS $D$ contains no information about $\tau$.  The token
		$\tau$ itself is chosen uniformly at random and is completely independent of
		the policy $\mathcal{P}_{B,b}$.  All remaining elements of $\mathsf{aux}$
		($S_0$, $C_0$, $C_F$, $\mathcal{C}$) are computed from $\tau$ and from
		independent random values ($\rho$, $k_2$, the PRF seeds).  Consequently, the
		entire challenge transcript is statistically independent of the bit $b$.  The
		adversary's probability of guessing $b$ in $H_1$ is therefore exactly $1/2$.
		
		Since $H_0$ and $H_1$ are computationally indistinguishable, the adversary's
		advantage in the real game $\mathsf{Priv}_{\mathcal{A}}^{\mathsf{policy}}$
		is bounded by a negligible function in $\lambda$.
	\end{proof}
	
	\subsection{Full Proof of On-chain Fair Exchange }
	\label{app:fair}
	
	\begin{theorem}[Fair Exchange with On‑Chain Audit]
		\label{thm:fair-full}
		Assume the zk‑SNARK is knowledge‑sound, the hash functions are
		collision‑resistant, the symmetric encryption provides ciphertext
		integrity, and the smart contract correctly executes the challenge‑response
		audit of Algorithm~\ref{alg:contract}.  Then for any PPT adversary
		$\mathcal{A}$ that corrupts either the seller or the buyer,
		$\Pr[\mathcal{A}\text{ wins in }\mathsf{Fair}_{\mathcal{A}}^{\mathsf{exch}}]
		\le \mathsf{negl}(\lambda)$.
	\end{theorem}
	
	\begin{proof}
		We show that the adversary can never make exactly one of the events
		$E_1$ (buyer obtains the decryption key) and $E_2$ (seller receives
		the payment) happen.
		
		\medskip
		\noindent \textbf{Case 1: Malicious buyer (honest seller).}
		To achieve $E_1 \land \neg E_2$ the adversary must obtain $k_2$ without
		the seller being paid.  The honest seller reveals $k_2$ on‑chain only
		after a valid buyer transaction is accepted.  A valid buyer transaction
		requires a proof $\pi$ that, by knowledge‑soundness, guarantees the
		prover holds a certified attribute set $\mathcal{A}_B$ with
		$|\mathcal{P}_B \cap \mathcal{A}_B| \ge t_b$ and a correctly reconstructed
		token $\tau'$.  Upon verification, the contract freezes the payment $v$.
		The honest seller then reveals the correct $k_2$, so $E_1$ becomes true
		and the payment is transferred, making $E_2$ true as well.  Hence
		$E_1 \land \neg E_2$ cannot occur.
		
		\medskip
		\noindent \textbf{Case 2: Malicious seller (honest buyer).}
		The seller aims for $E_2 \land \neg E_1$.  Three sub‑cases:
		\begin{enumerate}
			\item \emph{The seller does not reveal any key.}
			The challenge timer expires; the contract returns $v$ and transfers $d$
			to the buyer.  Hence $E_2$ is false and $E_1$ is false.
			\item \emph{The seller reveals a $k_2'$ that does not yield a valid
				decryption of $C_F$.}  The honest buyer detects the failure locally
			and submits a $\texttt{challenge}$ with the audit evidence described
			in Section~\ref{sec:audit-evidence}.  The contract verifies the Merkle
			proofs and the re‑encryption consistency check; because the seller
			supplied a wrong $k_2'$, the check fails for at least one challenged
			block.  Consequently the contract refunds $v$ to the buyer and
			transfers the collateral $d$ to the buyer, i.e., $E_2$ becomes false
			while $E_1$ remains false.  The adversary does not win.
			
			\item \emph{The seller reveals the correct $k_2$.}  Then $E_1$ is true
			and the payment is transferred, so $E_2$ is also true.
		\end{enumerate}
		The audit catches a cheating seller with overwhelming probability
		because $\mathcal{I}$ is chosen uniformly at random from $m$ blocks;
		the only way the seller could pass the audit while having tampered with
		the file is to forge a Merkle proof or break the signature on
		$\mathsf{root}$, which contradicts the collision‑resistance or
		unforgeability assumptions.
		
		Thus in every case $\Pr[E_1 \neq E_2] \le \mathsf{negl}(\lambda)$.
	\end{proof}

\end{document}